\documentclass[12pt, draftclsnofoot, onecolumn]{IEEEtran}
\usepackage{cite}
\usepackage{graphicx}
\usepackage{siunitx}
\usepackage{mathtools}
\usepackage{authblk}

\usepackage{amsmath}
\usepackage{amsthm}
\usepackage{amssymb}
\usepackage{dsfont}
\usepackage{epsfig}
\usepackage{color}
\newtheorem{proposition}{Proposition}

\usepackage{authblk}

\begin{document}
\title{Throughput of TCP over Cognitive Radio Channels \footnote{Preliminary versions of this paper were presented in Globecom 2015.}}
\author[1]{Sudheer Poojary}
\author[2]{Akash Agrawal}
\author[3]{Bhoomika Gupta}
\author[1]{Archana Bura}
\author[1]{Vinod Sharma}
\affil[1]{Department of ECE, Indian Institute of Science, Bangalore, India}
\affil[2]{Electrical and Electronics Engineering, BITS Pilani, K.K. BIRLA Goa campus, Goa, India}
\affil[3]{Electronics and Communication Engineering, Sardar Vallabhbhai National Institute of Technology, Surat, India}
\maketitle

\begin{abstract}
In this paper, we study the performance of a TCP connection over cognitive radio networks. In these networks, the network may not always be available for transmission. Also, the packets can be lost due to wireless channel impairments. We evaluate the throughput and packet retransmission timeout probability of a secondary TCP connection over an ON/OFF channel. We first assume that the ON and OFF time durations are exponential and later extend it to more general distributions. We then consider multiple TCP connections over the ON/OFF channel. We validate our theoretical models and the approximations made therein via ns2 simulations.

\end{abstract} 

\begin{keywords}
TCP, ON-OFF channels, cognitive radio.
\end{keywords}
\section{Introduction}
In this paper, we analyze performance of TCP connections over channels which are not always available for transmission. We call such channels ON-OFF channels. We limit ourselves to TCP Reno. On wireless channels misinterpretation of random packet losses as congestion losses can lead to poor performance of TCP. Besides these losses, in certain scenarios such as cognitive radio networks, the channel may not always be available \cite{Ghosh2014},\cite{Tadayon2015}. These lead to frequent timeouts and the TCP performance is adversely affected.

The primary motivation for this work is cognitive radio (CR) networks \cite{Mitola1999}. Cognitive radio is an emerging technology which intends to use the spectrum more efficiently. Huge portions of the spectrum are licensed but lay underutilized. Cognitive radio permits devices called secondary devices to utilize the channel when the primary user of the channel (the user to whom the channel is licensed) is not utilizing the channel. One can model this behaviour with an alternating ON-OFF channel, with ON periods corresponding to the primary being inactive and hence the channel being available for transmission by secondary users and OFF durations correspond to the primary busy periods where the channel is unavailable to the secondary for transmission. A practical example of TCP over an ON-OFF channel is cellular data boost \cite{Biglieri2012} in cellular networks where non-real time or delay tolerant traffic such as email, FTP etc, (which use TCP) can be  offloaded from the cellular network to a cognitive radio network of white space hot-spots to meet the QoS requirements of delay sensitive traffic. Besides CR networks, there are other networks where the links can be modeled as ON-OFF channels. Intermittent loss of connectivity also happens in cellular networks due to hand-offs, in mobile ad-hoc networks \cite{Mezzavilla2014} due to link failures, in satellite networks \cite{rfc2488} and in 802.11 networks due to collisions. 

The performance of TCP has been widely studied in the literature. In \cite{Mathis1997}, the authors develop a performance model for TCP and verify their model through simulations and measurements. In \cite{Padhye2000}, the authors provide an expression for TCP Reno throughput under Bernoulli packet losses using a Markov model. In \cite{Fu2003}, the authors study the performance of TCP over multi-hop wireless channel. They observe that link-layer contention causes packet drops and propose changes to the backoff mechanism of the link layer to improve TCP performance.

The performance of secondary users in a CR network has been studied in \cite{Ghosh2014} and \cite{Tadayon2015}. In \cite{Ghosh2014}, the authors compute the channel availability probability and the throughput for CR nodes in MIMO CR networks. In \cite{Tadayon2015}, the authors use Markov models to derive the stability condition for the secondary users in a multi-channel CR network. In \cite{Akyildiz2009}, the authors discuss properties and research challenges posed by cognitive radio. They suggest changes that need to be incorporated into the transport layer protocols for operation over CR channels. In \cite{Slingerland2007}, the authors compare the performance of TCP SACK, TCP New Reno and TCP Vegas over dynamic spectrum access links using ns2 simulations. In \cite{Felice2011}, the authors study the effect of spectrum sensing duration, primary user interference and channel bandwidth variation on different TCP variants using simulations. In \cite{Chen2012}, the authors use relay selection, power allocation and adaptive modulation and coding schemes to improve secondary user TCP performance over CR channels. In \cite{Wang2011}, the authors look at the impact of secondary sensing and primary activity on TCP throughput. In \cite{Kartheek2011}, the authors model the system as a M/G/1 queue with the primary users getting priority over the secondary users and provide expressions for throughput for data traffic and mean delay for voice traffic of secondary users. Transport layer protocols for cognitive radio have been developed in \cite{Sarkar2010, Kim2011, Chowdhury2013, Tsukamoto2015}. 

%\vspace*{-0.5mm}
From above we see that there is considerable literature on performance analysis of secondary users in a CR network. However most analytical results address the problem at the MAC layer ignoring the impact of TCP dynamics and the studies of TCP behaviour are mostly simulation-based. In this paper, we provide a theoretical model for TCP connections over a CR channel. We compute the throughput and probability of retransmission timeout for a secondary TCP connection. Our work complements \cite{Kartheek2011}. In \cite{Kartheek2011}, the authors consider the case where ON and OFF durations are of the order of round trip time (RTT)  and hence they ignore TCP timeouts. In our work, we consider the case where the ON and OFF durations are larger than RTT where the effect of TCP timeouts cannot be ignored. Such a scenario can often happen in CR networks and then the timeouts can significantly affect the TCP throughput. We note that \cite{Wang2011} also considers the impact of primary user activity on TCP throughput due to timeouts. However they consider a slotted model, where the primary user if active at the beginning of the slot stays active for the entire slot and if inactive stays inactive for the entire slot. We do not consider a slotted system and consider more general primary user behaviour.

%\vspace*{-0.5mm}
This paper is organized as follows. In Section \ref{sec:systemModel}, we describe our system model. In Section \ref{sec:perfAnalysis}, we develop a Markov model for TCP behaviour in an ON-OFF channel with exponential ON and OFF periods. In Section \ref{sec:phase_type}, we develop models for ON and OFF periods with more general distributions using regenerative theory. In Section \ref{sec:multipleTCP}, we consider the system with multiple secondary TCP connections. Section \ref{sec:conclusion} concludes this paper.
\section{System Model}
\label{sec:systemModel}

\begin{figure}[t]
  \centering
  \includegraphics[scale = 0.7]{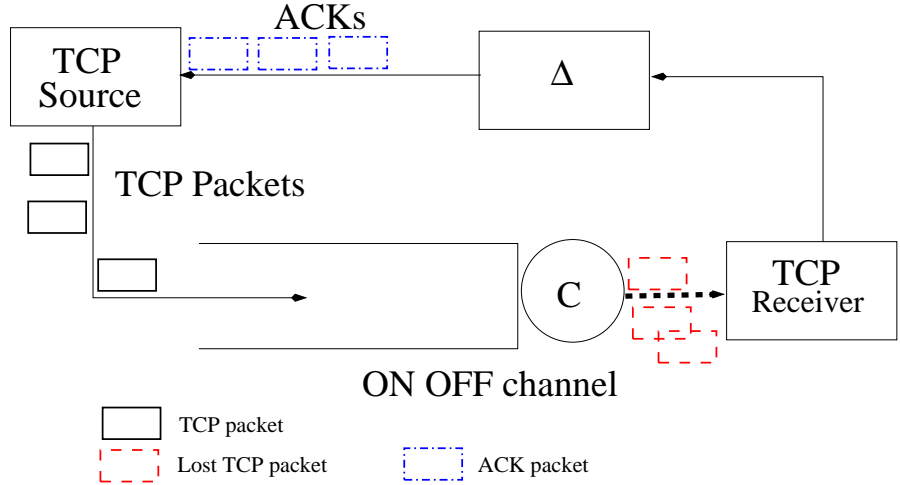}
  \caption{TCP connection through ON/OFF channel}
  \label{fig:systemModel1}
%  \vspace*{-5mm}
\end{figure}   

We consider a single TCP New Reno flow going through an ON-OFF channel as shown in Figure \ref{fig:systemModel1}. In the Figure, the packets in red are TCP packets lost due to the channel going OFF. As the channel alternates between ON and OFF, the TCP receiver receives packets intermittently. The TCP flow could be of a secondary device in a cognitive radio network using opportunistic spectrum sharing so that the OFF periods correspond to the primary using the channel and the ON periods correspond to the time when the channel is available for the secondary user. The overall RTT of the TCP flow is fixed and equals $R$ seconds. (This corresponds to negligible queuing. We will extend this assumption later.) Any transmission (of a window load of packets) when the channel is ON goes through, although the packets may experience transmission error. The next window of packets will be transmitted after RTT, i.e., $R$ secs. However a transmission attempted during an OFF period results in loss of all the packets of that window and hence causes a retransmission timeout (RTO) and the TCP source has to wait for its RTO timer to expire before it can attempt the next retransmission. It is possible that the channel becomes OFF while a secondary transmission is going on. Such an event is likely if the average ON and OFF periods are of the order of RTT or lesser. This case has been studied in \cite{Kartheek2011}. Here we consider the case where the average ON and OFF periods are larger than the RTT of the flow with high probability. 

If during the ON period, the transmission error probability is small, after an error the next transmitted packets (i.e., packets in the window succeeding the packet in error) may be received successfully causing transmission of duplicate ACKs. In response to these losses, the secondary TCP reduces its window size. In Figure \ref{fig:systemModel2}, we show the window size evolution for a TCP flow over an ON-OFF channel. In the ON period, the TCP window size increases, however random losses over the channel cause reductions in the window size even during the ON periods. 

Let $S_k \in \{0,1\}$ be the channel state at the $k^{th}$ transmission of a window load of packets, where $1$ corresponds to the channel being ON and $0$ corresponds to the channel being OFF. Let $J_k \in \mathcal{D}$ denote the duration between the $k^{th}$ and the $(k+1)^{st}$ transmission of a window load of packets. The duration between successive transmissions in the ON period is equal to the RTT, $R$, of the flow. However, on encountering a timeout, TCP uses a binary exponential back-off strategy. The first RTO is set to $M = \max \{R, T_{min}\}$, where $T_{min}$ is the minimum value that a timeout duration can be. If a failed transmission (one which leads to a timeout) is immediately followed by another, TCP doubles the timeout duration. The timeout duration is bounded above by $T_{max}$. Therefore $\mathcal{D} = \{R, M, 2M, 4M, \cdots, T_{max} \}$. The TCP timeout mechanism is  illustrated in Figure \ref{fig:systemModel2}.

\begin{figure}[t]
  \centering
  \includegraphics[scale = 0.2]{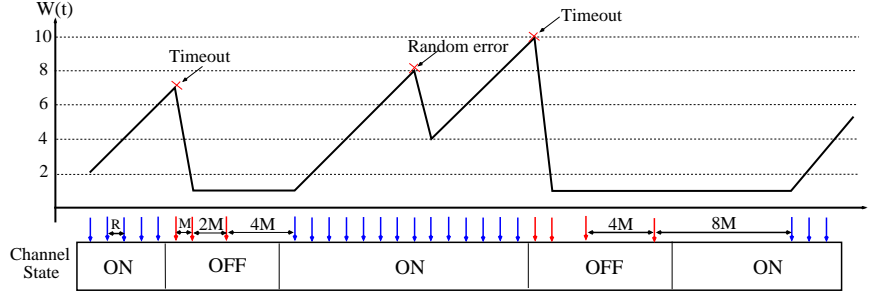}
  \caption{Cognitive radio channel with one secondary TCP source.}
  \label{fig:systemModel2}
%  \vspace*{-5mm}
\end{figure}

Let $W_k$ be the window size at the end of the $k^{th}$ transmission of a window of packets and $H_k$ be the value for the corresponding slow start threshold of TCP. The window evolution of TCP New Reno is as follows. If there is no packet loss between the $k^{th}$ transmission and the $(k+1)^{st}$ transmission of windows, we have
\begin{equation}
\begin{split}
\label{eqn:Reno_CA}
W_{k+1} &= W_k + 1, \text{ if }  W_k \geq H_k  \\
&= 2 W_k,    \text{ if } W_k < H_k. 
\end{split}
\end{equation}
If there is a packet loss, TCP retransmits the lost packet and reduces the window size. If TCP detects the loss through duplicate ACKs, then it reduces the window size by half and sets the slow start threshold to that value. This is called recovery through \textit{fast retransmit}\cite{rfc2581}. In that case,
\begin{equation}
\label{eqn:Reno_FR}
W_{k+1} = \frac{W_k}{2}, \text{ } H_{k+1} = \frac{W_k}{2}.
\end{equation}
If the loss is detected through a timeout (this will happen probably because channel is OFF), we have
\begin{equation}
\label{eqn:Reno_TO}
W_{k+1} = 1, \text{ } H_{k+1} = \frac{W_k}{2}.
\end{equation}
The TCP window size is usually restricted by the buffer size available at the receiver. Considering this, the window size $W_k$ and the threshold $H_k$ are restricted to $W_{max} < \infty$.

In the rest of the paper, we develop theoretical models for TCP over an ON-OFF channel and compute the probability of retransmission timeout and the throughput and compare our results to ns2 simulations. We then show how these results can be used when there are multiple TCP connections.
\section{Analysis for Exponential ON-OFF}
\label{sec:perfAnalysis}
We assume that the OFF and ON periods are i.i.d. exponential with parameters $\lambda_0$ and $\lambda_1$ respectively. This is reasonable as the PU activity is usually modeled as exponential \cite{Akyildiz2009}. We will generalize these assumptions in Section \ref{sec:phase_type}.

We denote the state of the system at the beginning of the $k^{th}$ transmission of a window workload by $(S_k, J_k, W_k, H_k)$. Since we assume that the ON and OFF durations are exponential, the process, $\{(S_k, J_k, W_k, H_k)\}$ forms a finite state, discrete time Markov chain. In Figure \ref{fig:markovChain_tpm}, we illustrate the single-step transitions from generic ON and OFF states. 

Let $S(t)$ be the state of the channel at time $t$ with $S(t) = 0$ if channel is OFF and $1$ if it is ON. In the ON period, packets can be dropped due to losses on the wireless channel with probability $p$ independently of others. Let $P_t(i,j)$ be the probability of channel state $S(t)$ at time $t$ being in state $j$ given that it was in state $i$ at time $0$. Once we know $P_t(i,j)$, we can find the transition probabilities for the Markov chain $\{(S_k, J_k, W_k, H_k)\}$. Proposition \eqref{prop:transition_probs} provides $P_t(i,j)$ explicitly via renewal theory.

\begin{figure}[t!]
  \centering
\includegraphics[scale=0.6]{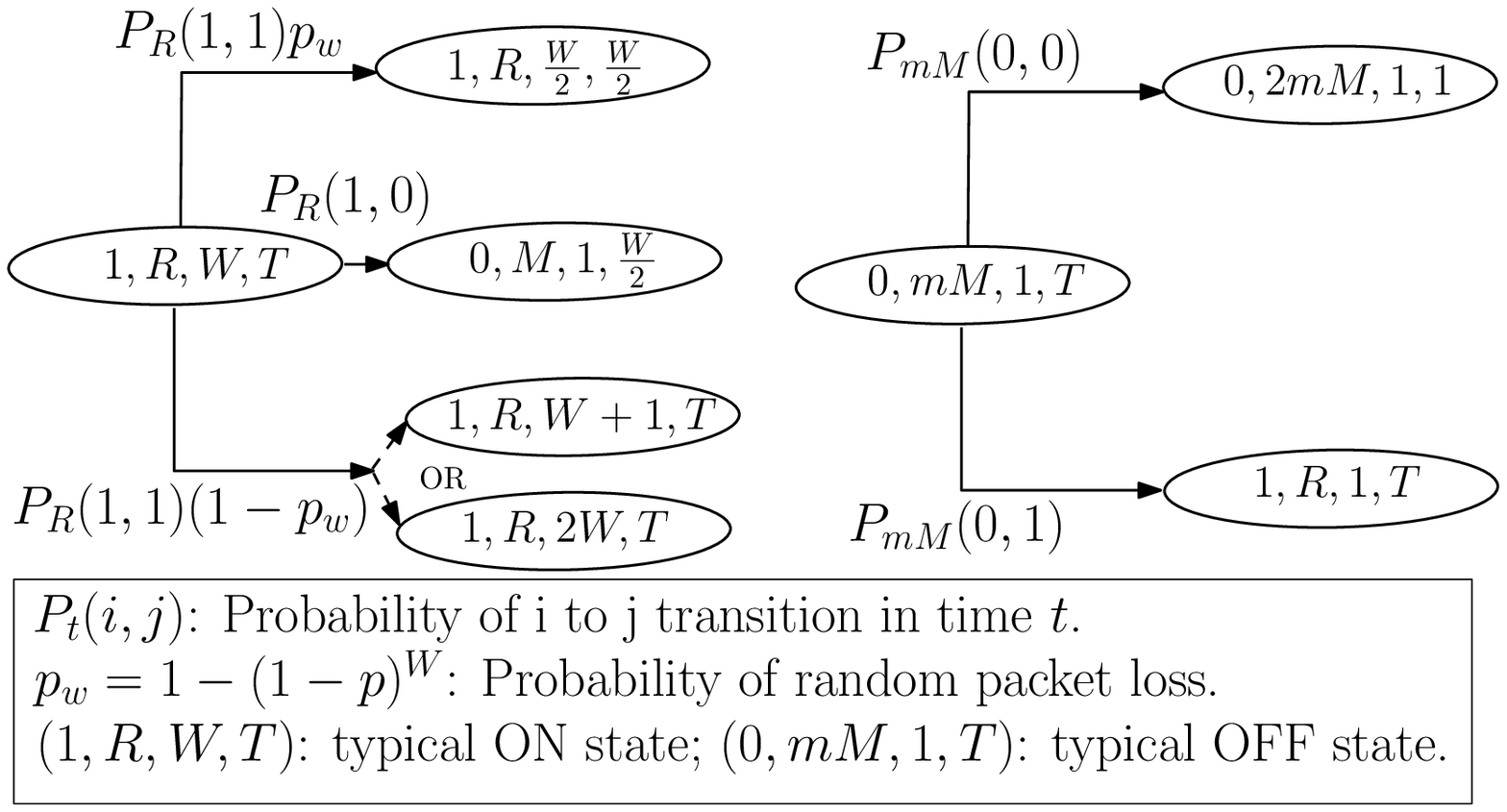}
  \caption{Single-step transitions for the $\{S_k, J_k, W_k, H_k\}$ Markov chain.}
  \label{fig:markovChain_tpm}
%  \vspace*{-5mm}
\end{figure}

\begin{proposition}
\label{prop:transition_probs}
We have for $i \neq j$,
\begin{equation}
\label{eqn:b2b}
P_t(i,i) = \frac{\lambda_j + \lambda_i e^{-(\lambda_0 + \lambda_1)t}}{\lambda_0 + \lambda_1}.
\end{equation}
\end{proposition}

\begin{proof}
Let $X_i$ be exponentially distributed with parameter $\lambda_i$, $i = 0,1$ and $X_0$ independent of $X_1$. Let $Y = X_0 + X_1$. Let us denote the density function of $Y$ by $f_Y$. By renewal theory arguments \cite{Wolff1989}, we have,
\begin{equation}
\label{eqn:B2B_renewal_arg}
P_t(0,0) = \mathbb{P}(X_0 > t) + \int_{0}^{t}  P_{t-x}(0,0) f_Y(x) dx.
\end{equation}
Taking Laplace transform, we have
\begin{equation}
\label{eqn:B2B_Laplace}
\hat{P}_s(0,0) = \frac{1}{\lambda_0 + s} + \hat{P}_s(0,0) \hat{f}_Y(s).
\end{equation}
The Laplace transform, $\hat{f}_Y$ of $f_Y(t)$ is given by,
\begin{equation}
\label{eqn:Laplace_sum_expo}
\hat{f}_Y(s) = \frac{\lambda_0 \lambda_1}{(\lambda_0 + s) (\lambda_1 + s)}.
\end{equation}
From equations \eqref{eqn:B2B_Laplace} and \eqref{eqn:Laplace_sum_expo}, we have
\begin{equation}
\label{eqn:Laplace_P_B2B}
\hat{P}_s(0,0) = \frac{\lambda_1}{\lambda_0 + \lambda_1} \frac{1}{s} + \frac{\lambda_0}{\lambda_0 + \lambda_1} \frac{1}{s + \lambda_0 + \lambda_1}.
\end{equation}
Taking inverse Laplace transform, 
\begin{equation}
\label{eqn:P_B2B}
P_t(0,0) = \frac{\lambda_1}{\lambda_0 + \lambda_1} + \frac{\lambda_0}{\lambda_0 + \lambda_1} e^{-(\lambda_0 + \lambda_1)t}.
\end{equation}
Similarly we can derive expression for $P_t(1,1)$.
\end{proof}

The Markov chain, $\{(S_k, J_k, W_k, H_k)\}$ is finite. Therefore it has at least one positive recurrent communicating class. Suppose the channel is ON. Then once the ON period ends, an OFF period large enough so that we have at least two timeouts, followed by an ON period, large enough to accommodate one transmission attempt, ensures that the state $(1, R, 1, 1)$ is hit. Similarly, when the channel is OFF, an OFF period large enough so that we have a total of at least two timeouts, followed by an ON period, large enough to accommodate one transmission attempt, ensures that the state $(1, R, 1, 1)$ is hit. Thus, the state $(1, R, 1, 1)$ can be reached from any state in the state space with positive probability. Therefore the state $(1, R, 1, 1)$ is positive recurrent and any state that can be reached from $(1, R, 1, 1)$ is positive recurrent and the remaining states are transient. Also if the probability of packet loss during the ON period is greater than $0$ (which happens if we assume that packets are in error, independently when transmitted during ON periods, e.g., for wireless channels), the state $(1, R, 1, 1)$ has a self loop. Therefore the Markov chain is aperiodic and has a unique stationary distribution $\pi$. Also, starting from any initial state the chain converges exponentially to the stationary distribution in total variation.

%\subsection{Probability of retransmission timeout and throughput}
%\label{subsec:inspection paradox}
The probability of retransmission timeout, $P_o$, i.e., the fraction of packets that are timed out is 
\begin{equation}
\label{eqn:PKT_RTO}
P_o = \frac{\mathbb{E}_{\pi}[1_{\{S = 0\}}]}{\mathbb{E}_{\pi}[W]},
\end{equation}
where $1_{ \{ A \}}$ is an indicator function of set $A$ and $E_{\pi}$ denotes mean under stationarity. The throughput, $\lambda$ (in packets/sec) of the TCP connection, using Palm calculus \cite{Asmussen}, is
\begin{equation}
\label{eqn:T_palmcalculus}
\lambda = \frac{\mathbb{E}_{\pi}[W 1_{ \{S = 1 \} }]}{\mathbb{E}_{\pi}[D]}.
\end{equation}
%In the next section, we compare $P_{o}$ and $\lambda$ obtained using the above model with ns2 simulations.

\subsection{Extension to Channels with Non-negligible Queuing}
The above analytical model assumes that the round trip time for the secondary TCP connection is constant. If the queuing is non-negligible, the model may not be accurate. In that case, we can use the above model with a minor modification. In the state space of the process $\{S_k, J_k, W_k, H_k\}$, if the channel is ON at the end of the $k^{th}$ transmission, i.e., if $S_k = 1$, we set $J_k = \max\{\Delta, \frac{W_k}{\mu}\}$, where $\Delta$ is the constant component of the round trip time which includes the propagation delay and processing delays at the nodes and $\mu$ is the link speed of the channel in packets/second. Such an approximation has been used before in \cite{Bonald1999}, \cite{Blanc2009}. Our simulation results below justify the approximations made in this analytical model.

\subsection{Simulation Results}
\label{sec:simulationResults}
We now compare the probability of timeout and throughput obtained from our model with ns2 simulations. We have modified ns2 code so as to simulate an ON-OFF channel. We generate a sequence of alternate ON and OFF periods and drop all packets that arrive in the OFF period. The packet sizes are $1050$ bytes. We set the link speed of the ON-OFF channel to $5$ Mbps. The other links that the flow traverses have link speeds $1$ Gbps. In practice, the ON-OFF channel could be a wireless link connecting a wireless device to an access point or a base station. The base station/access points are then connected to the Internet through well-provisioned optical fiber links.

%\subsection{Effect of RTT and average durations of the ON and OFF periods}
We denote the fraction of time that the channel is OFF by $\alpha$. We plot the probability of retransmission timeout (RTO), $P_o$ and the secondary TCP throughput in Figure \ref{fig:Bexp_Iexp_eff_RTT}.
%Figures \ref{fig:Bexp_Iexp_pkt_rto_eff_RTT} and \ref{fig:Bexp_Iexp_goodput_eff_RTT}. 
We set $\alpha = 1/3$ and vary the average channel OFF duration, $\mathbb{E}[Y_{off}]$. The probability of packet transmission error in the ON period is set to $0.01$. The maximum window size, $W_{max}$ is $100$ packets. We see that as RTT increases, the probability of retransmission timeout increases and throughput decreases. Our model results match well with ns2 simulations and the errors are less than $5\%$ in most cases. The errors are larger when the average ON and OFF durations are of the order of RTT. However, even for these cases the errors are less than $11\%$. Our analytical model results show that in an ON-OFF channel we can not approximate TCP throughput by multiplying the throughput expressions for TCP Reno (which may be found in \cite{Padhye2000}) by $(1 - \alpha)$ where $\alpha$ is the fraction of time that the channel is busy.

%\begin{figure}
%  \centering
%  \includegraphics[scale=0.15]{Bexp_Iexp_pkt_rto_eff_RTT.eps}
%  \caption{Effect of RTT and $E[B]$ on probability of RTO.}
%  \label{fig:Bexp_Iexp_pkt_rto_eff_RTT}
%\end{figure}   
%
%\begin{figure}
%  \centering
%  \includegraphics[scale=0.16]{Bexp_Iexp_goodput_eff_RTT.eps}
%  \caption{Effect of RTT and $E[B]$ on TCP throughput}
%  \label{fig:Bexp_Iexp_goodput_eff_RTT}
%\end{figure}   

\begin{figure}
\centering
\begin{tabular}{c}
\includegraphics[scale=0.24, trim = 65 10 140 5, clip=true]{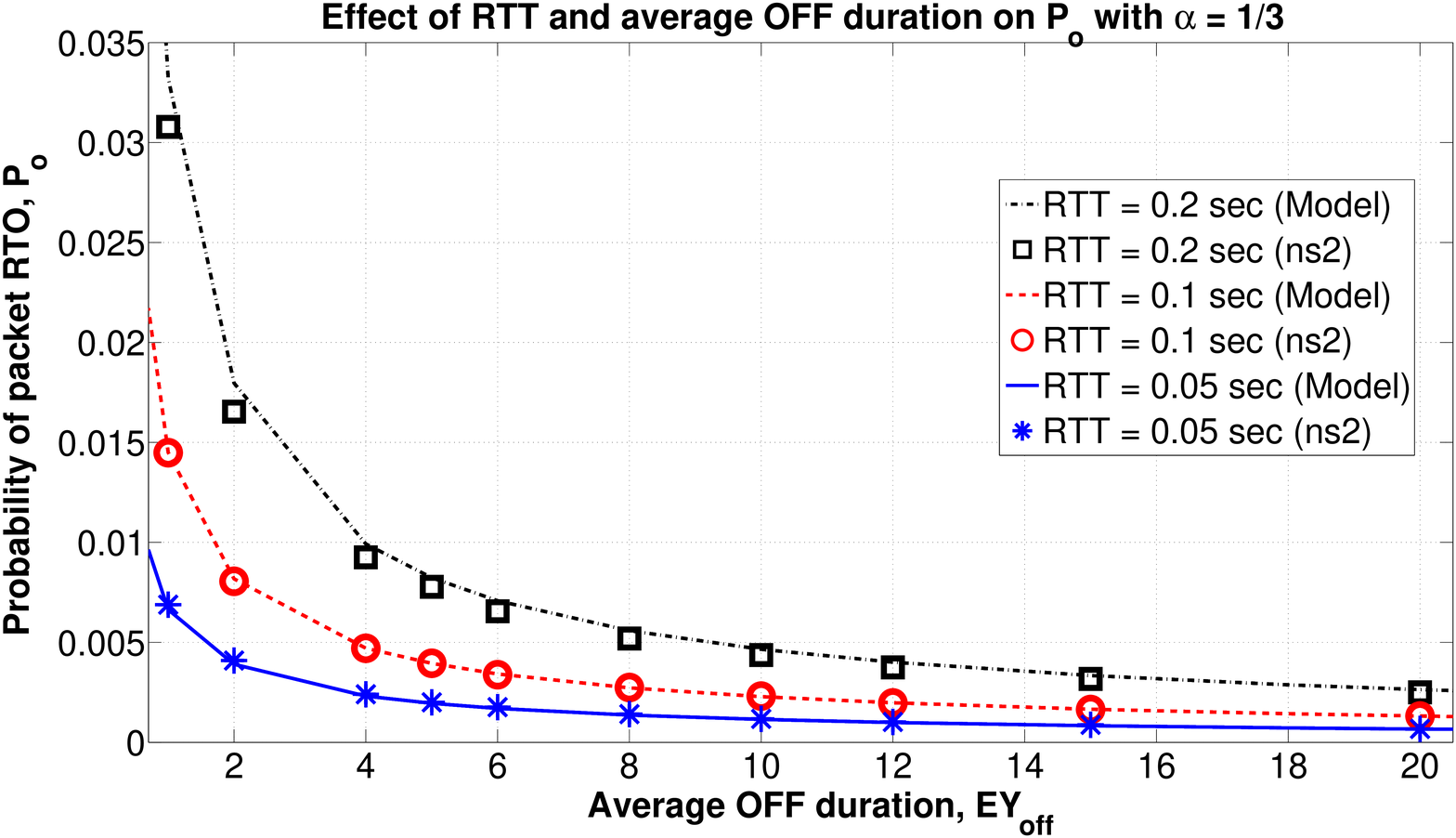} \\ \\ \\
\includegraphics[scale=0.24, trim = 65 10 140 5, clip=true]{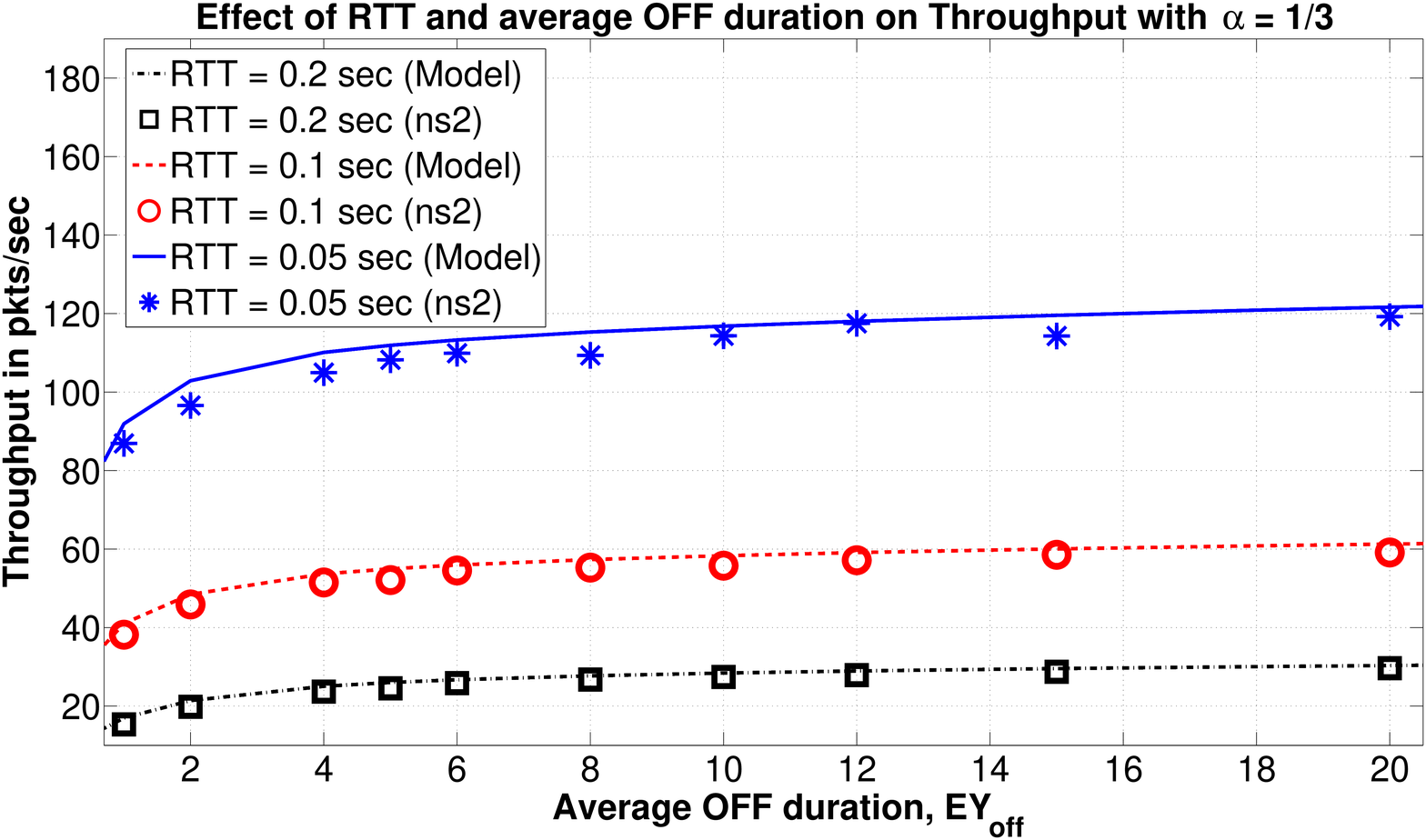}
\end{tabular}
\caption{Effect of RTT and $\mathbb{E}[Y_{off}]$ on probability of RTO and TCP throughput.}
\label{fig:Bexp_Iexp_eff_RTT}
%\vspace*{-4mm}
\end{figure}
%\subsection{Effect of random losses}
In Figure \ref{fig:Bexp_Iexp_eff_per}, 
% \ref{fig:Bexp_Iexp_pkt_rto_eff_per} and \ref{fig:Bexp_Iexp_goodput_eff_per}, 
we see the effect of packet error probability $p$ on the probability of retransmission timeout and the throughput of the flow. The round trip time is $0.1$ seconds. We see that for fixed $\alpha$ and RTT, $R$, the probability of retransmission timeout decreases and throughput increases with increase in the average OFF duration. Our model results match well with ns2 simulations and the errors are less than $8\%$. 

\begin{figure}
\centering
\begin{tabular}{c}
\includegraphics[scale=0.24, trim = 65 10 140 5, clip=true]{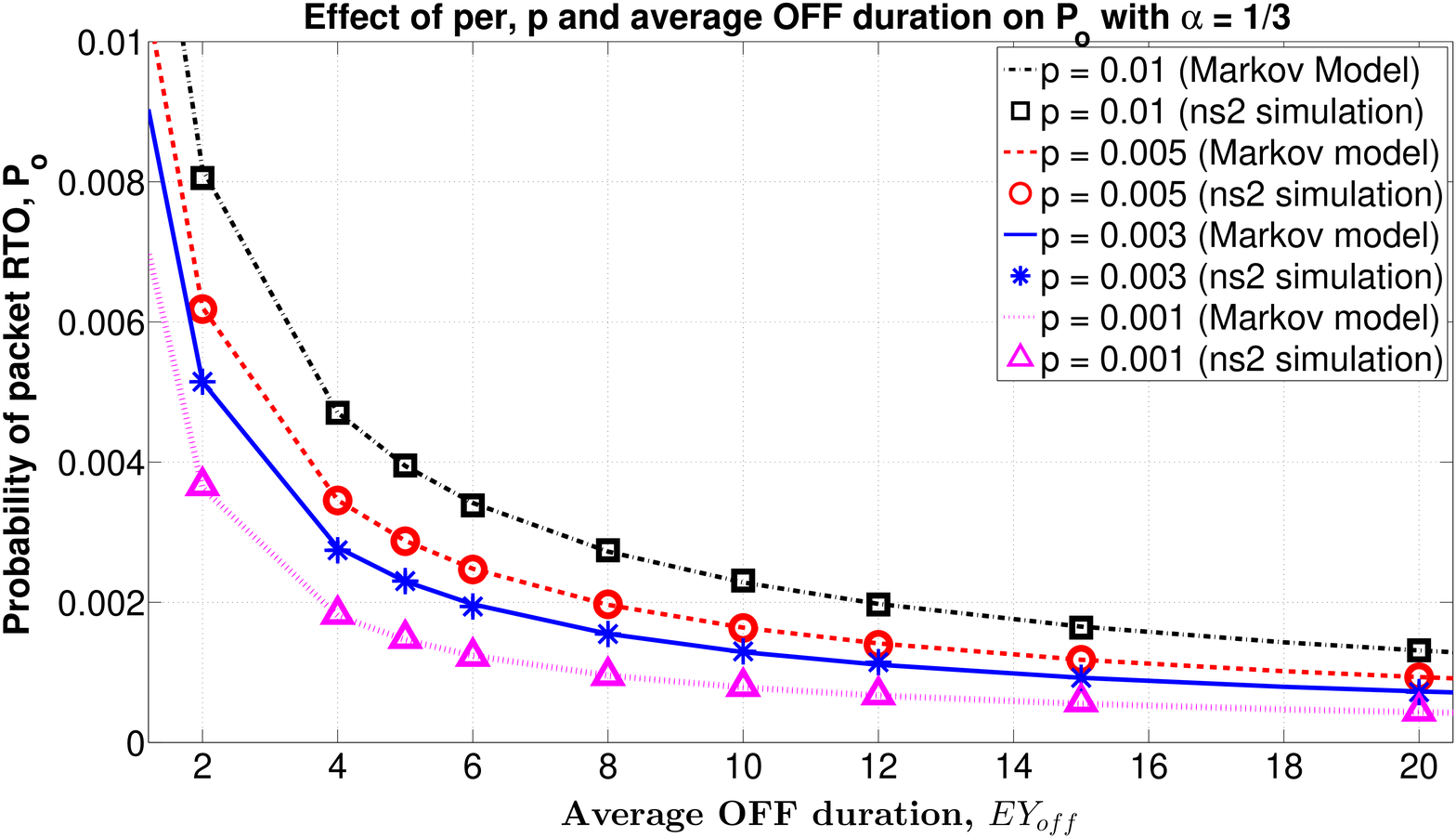} \\ \\ \\
\includegraphics[scale=0.24, trim = 65 10 140 5, clip=true]{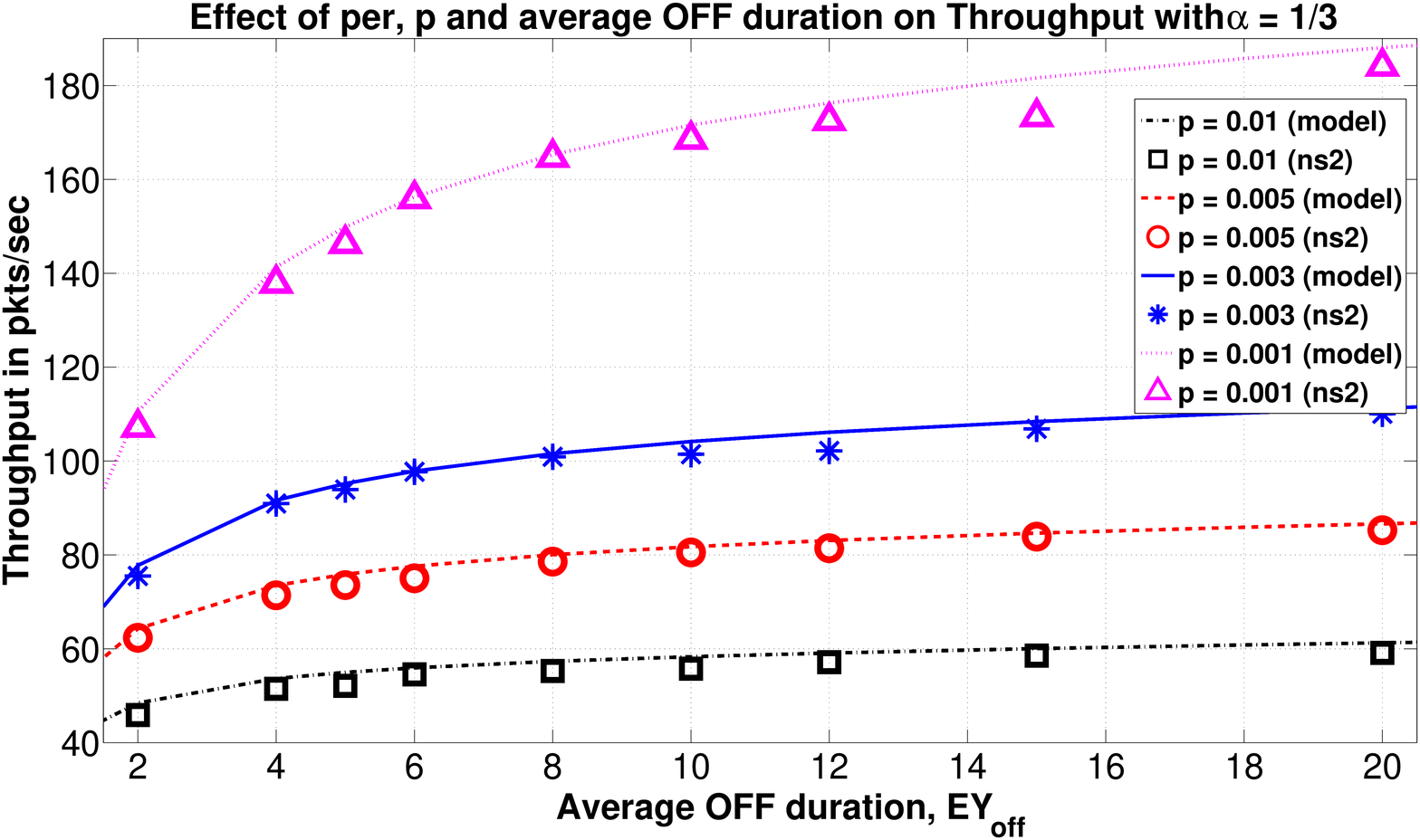}
\end{tabular}
\caption{Effect of packet error rate and $\mathbb{E}[Y_{off}]$ on TCP performance.}
\label{fig:Bexp_Iexp_eff_per}
%\vspace*{-2mm}
\end{figure}

%\begin{figure}
%  \centering
%  \includegraphics[scale=0.16]{Bexp_Iexp_pkt_rto_eff_per.eps}
%  \caption{Effect of packet error probability and $E[B]$ on probability of RTO}
%  \label{fig:Bexp_Iexp_pkt_rto_eff_per}
%\end{figure}   
%
%\begin{figure}
%  \centering
%  \includegraphics[scale=0.16]{Bexp_Iexp_goodput_eff_per.eps}
%  \caption{Effect of packet error probability and $E[B]$ on TCP throughput}
%  \label{fig:Bexp_Iexp_goodput_eff_per}
%\end{figure} 
In Figure \ref{fig:Bexp_Iexp_X_5M_1M}, 
we show the effect of link speed on the probability of retransmission timeout and the throughput of the TCP flow. The RTT is $0.1$ seconds and CR channel link speeds are set to $5$ Mbps and $1$ Mbps. We see that our model approximations are reasonable, the errors are less than $10\%$.
\begin{figure}
\centering
\begin{tabular}{c}
\includegraphics[scale=0.24, trim = 65 10 140 5, clip=true]{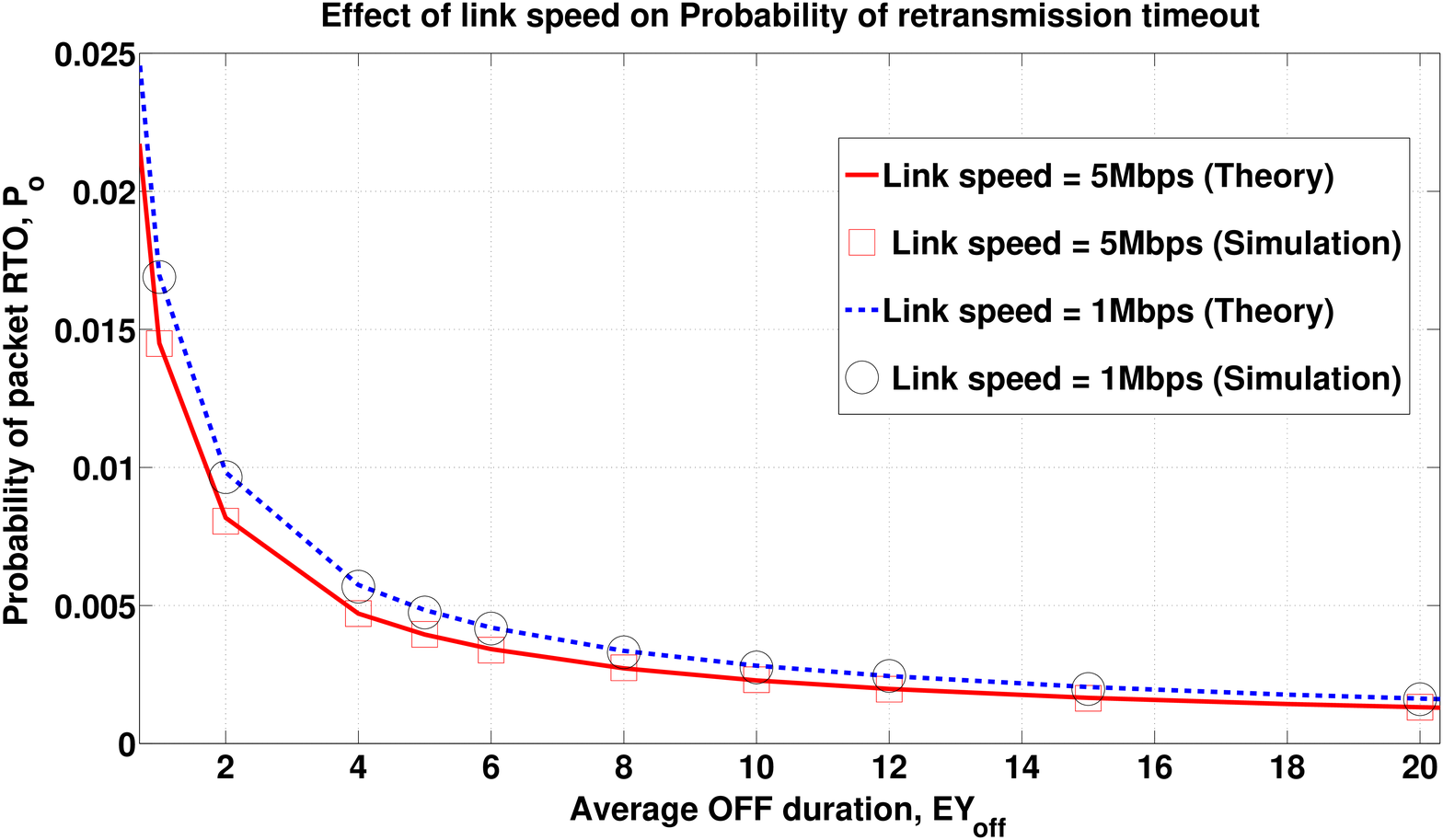} \\ \\ \\
\includegraphics[scale=0.24, trim = 65 10 140 5, clip=true]{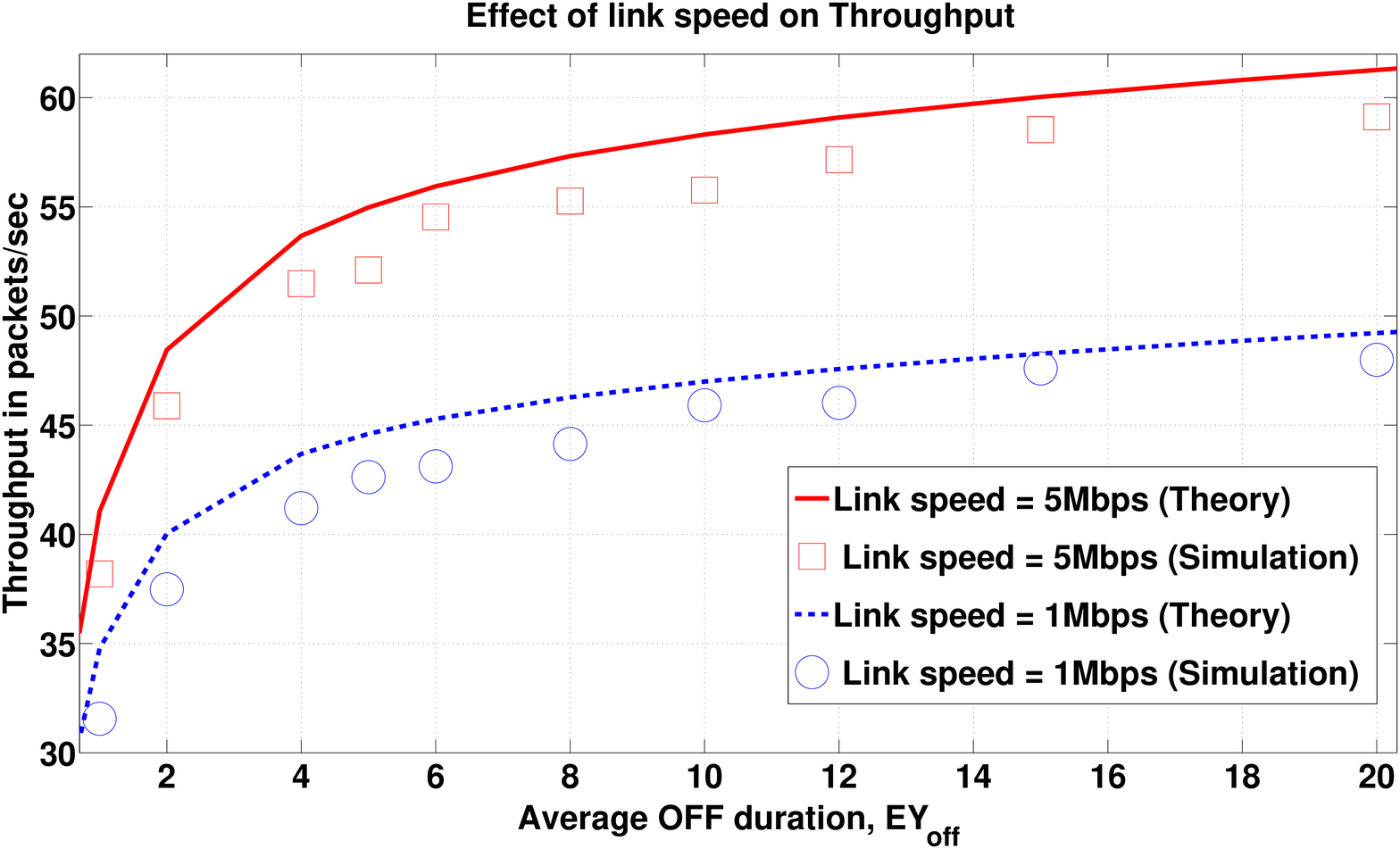}
\end{tabular}
\caption{Effect of link speed and $\mathbb{E}[Y_{off}]$ on probability on TCP performance.}
\label{fig:Bexp_Iexp_X_5M_1M}
%\vspace*{-4mm}
\end{figure}
\section{Phase-type Distributions for ON/OFF Periods}
\label{sec:phase_type}
We have considered the case where the ON and OFF periods are exponentially distributed. In \cite{Feldmann1997}, we see that many quantities that characterize network performance, for example file lengths, call holding times, intervals between requests in Internet traffic are not exponential but can be modelled by mixtures of exponential distributions. We now extend our model to the case where both ON and OFF periods are phase-type \cite{Asmussen}. Phase-type distributions generalize exponential distributions (and include mixtures of exponentials) and can approximate any distribution on $\mathbb{R}^{+}$ arbitrarily closely \cite{Asmussen}. As before, the sequence of ON periods and OFF periods are both  i.i.d. and are independent of each other. Also, as before the state of the system at the beginning of the $k^{th}$ transmission of a window workload is denoted by $(S_k, J_k, W_k, H_k)$. 

When ON and OFF periods are phase-type, the channel state $S(t)$ can be taken as a finite state Markov chain with state space $\mathcal{X}_0 \cup \mathcal{X}_1$. When $S(t) \in \mathcal{X}_0$ then the channel is OFF and when $S(t) \in \mathcal{X}_1$ then it is ON. Also, the process $(S_k, J_k, W_k, H_k)$ is a discrete time Markov chain (embedded in $S(t)$). Let $P_t(i,j)$ be the probability of the channel state $S(t)$ at time $t$ being $j$ given that it was in state $i$ at time $0$. Once we know $P_t(i,j)$, we can find the transition probabilities for the Markov chain $\{(S_k, J_k, W_k, H_k)\}$. 

Suppose $Q$ is the transition rate matrix for the CTMC $\{S(t)\}$, and let $\psi_t := \{ \psi_t(i) : i \in \mathcal{X}_0 \cup \mathcal{X}_1 \}$ denote the probability distribution of $S(t)$ at time $t > 0$. It satisfies
\begin{equation}
\psi_t = \psi_0 e^{Qt}, 
\end{equation}
where $\psi_0$ is the initial distribution. Also, $P_t = e^{Qt}$. It can be computed using the $expm$ command in MATLAB or by using eigenvalue decomposition techniques \cite{Perko2001}. When the state space of $S(t)$ is large, the exact computation of  $e^{Qt}$  may be difficult. We may then use some approximations \cite{Reibman1988}.

The Markov chain $\{(S_k, J_k, W_k, H_k)\}$ is finite. Therefore it has at least one positive recurrent communicating class. Suppose the channel is ON. Then once the ON period ends, an OFF period large enough so that we have at least two timeouts, followed by an ON period, large enough to accommodate one transmission attempt, ensures that the set of states $\{(i, R, 1, 1): i \in \mathcal{X}_1\}$ is hit. Similarly, when the Markov chain is in an OFF state, an OFF period large enough so that we have a total of at least two timeouts, followed by an ON period, large enough to accommodate one transmission attempt, ensures that the set of states $\{(i, R, 1, 1): i \in \mathcal{X}_1\}$ is hit. Since the ON periods are phase type, there is a non-zero probability that for any $i \in \mathcal{X}_1$, the state $(i, R, 1, 1)$ is visited by the Markov chain. Thus, the state $(i,R,1,1)$ for any $i \in \mathcal{X}_1$ is a positive recurrent state. If the probability of packet loss during the ON period is greater than $0$, then the state $(i,R,1,1)$ also has a self-loop. Therefore the Markov chain is aperiodic and has a unique stationary distribution $\pi$. Also, starting from any initial state the chain converges exponentially to the stationary distribution in total variation.

The steady state probability distribution, $\pi$, of the Markov chain, $\{(S_k, J_k,$ $W_k, H_k)\}$ is computed by repeated iteration of $\pi_{k+1} = \pi_k P$, where $P$ is the transition probability matrix, starting from some initial $\pi_0$. The size of the state space of the Markov chain is $O(| \mathcal{X}_0 + \mathcal{X}_1| |\mathcal{D}| W_{max}^2)$. The number of transitions from a state are of the order $O(2| \mathcal{X}_0 + \mathcal{X}_1|)$. Thus each iteration of $\pi_{k+1} = \pi_k P$ requires $O(2 | \mathcal{X}_0 + \mathcal{X}_1|^2 |\mathcal{D}| W_{max}^2)$ computations. From our earlier discussion, these iterations converge exponentially to the stationary distribution.
% State space size for Erlang-3, Erlang-3, with W_{max} = 100 and link speed = 5Mbps was 12216, theoretical upper bound is 1200000.

We now develop theoretical results for (a) phase-type OFF and general ON periods under certain conditions and (b) phase-type ON and general OFF periods under certain conditions. We will show stationarity of the process $\{(S_k, J_k, W_k, H_k)\}$  in these cases. The stationary behaviour can then be used to compute the throughput and the probability of timeout for a TCP flow on an ON-OFF channel.

\subsection{Phase-type OFF Periods}
In this section, we consider the case when OFF periods are i.i.d. phase-type while the ON periods are more general. Let us denote by $U_k$ the time 
when the secondary TCP makes the $k^{th}$ transmission attempt of a window load of packets. Let $U_0 = 0$. Therefore $U_k = \sum_{n=1}^{k-1} J_n$, where $J_n$ is the time duration between the $n^{th}$ and $(n+1)^{st}$ transmission attempt. Let $\{X(t)\}$ be a CTMC with finite state space $\{1\} \cup \mathcal{X}_0$, where $1 \notin \mathcal{X}_0$. The set $\mathcal{X}_0$ is irreducible and state $1$ is an absorbing state reachable from every state in $\mathcal{X}_0$. The sojourn times in each state $i \in \mathcal{X}_0$ are exponential with parameter $\lambda_i$, $0 < \lambda_i < \infty$. Let $\lambda \triangleq  \max_{i \in \mathcal{X}_0} \lambda_i < \infty$. We define the channel state process $S(t)$ as follows. Suppose $S(t)$ enters ON state at time $t$ (denoted by  $S(t) = 1$). It will stay in ON state with a general distribution ($Y_{on}$ will denote a random variable with that distribution). At time $t +   Y_{on}$, it will enter the OFF period. The duration of the OFF period is given as follows. Let $X(0) = i_0 \in \mathcal{X}_0$ where $i_0$ is a fixed  state of $\{X(t)\}$. The OFF period of $S(t)$ will equal the time $\{X(t), t \geq 0\}$ takes to reach state $1$. Let this time be $X$. Then $S(t)$  stays in OFF state till time $t + Y_{on} + X$ and then switches to ON state . At time $t + Y_{on} + s$, $S$ takes the value $X(s)$, for $0 \leq s \leq  X$. The ON-OFF periods of $S(t)$ alternate with durations independent of each other with distributions specified above. Let $Q$ be the generator matrix  for the CTMC $X(t)$ when in $\mathcal{X}_0$. By uniformization \cite{Stewart2009}, before $X(t)$ gets absorbed, the transition function of $\{X(t)\}$ is, 
\begin{equation}
\label{eqn:ctmc_trans_prob}
P_t(i,j) = \sum_{n=0}^{\infty} e^{-\lambda t} P^{n}_{i,j} \frac{(\lambda t)^n}{n!},
\end{equation}
where $P_{ij} = \frac{Q_{ij}}{\lambda}$ for $i \neq j$ and $P_{ii} = 1 - \sum_{j \neq i} \frac{Q_{ij}}{\lambda}$.

Consider the state $(S_k = i_0, J_k = 2M, W_k = 1, H_k = 1)$ where $i_0 \in \mathcal{X}_0$. Once the process $\{(S_k, J_k, W_k, H_k)\}$ hits the state, $(i_0, 2M, 1, 1)$, the future process evolution is independent of the past. Thus, $\{(S_k, J_k,  W_k, H_k)\}$ is regenerative. Let $N$ denote the number of transmission attempts made by TCP in a regeneration cycle. We have the following result.

\begin{proposition}
\label{thm:phaseOFF_stationary}
If the ON periods are i.i.d. with 
\begin{equation}
\label{eqn:NBU}
\mathbb{P}(Y_{on} \leq s + R| Y_{on} \geq s) \geq \epsilon_0 > 0
\end{equation} 
and $0 < \mathbb{P}(Y_{on} \leq R) < 1$, then for all $\alpha > 0$, $\mathbb{E}[N^{\alpha}] < \infty$ and $N$ has a finite moment generating function (mgf) in a neighbourhood of $0$. Also $N$ is aperiodic and the stochastic process $\{(S_k, J_k, W_k, H_k)\}$ converges in total variation, exponentially to its unique stationary distribution
\begin{equation}
\label{eqn:regen_eqn_phaseOFF}
\begin{aligned}
\mathbb{P}_{\pi}((S_k, &  J_k, W_k, H_k)  \in A ) \\
 & = \frac{\mathbb{E} [ \sum_{k=0}^{N} \mathbf{1}_A (S_k, J_k, W_k, H_k)]}{\mathbb{E}[N]},
\end{aligned} 
\end{equation}
where $k=0$ is a regeneration epoch.
\end{proposition}

\begin{proof}
We first prove that from any arbitrary state in the OFF period the process can visit the ON state in one step with probability $\geq \epsilon_1 > 0$.

\textbf{Step 1: } Since $\mathcal{X}_0$ is irreducible and $1$ is reachable from all states, for all $i \in \mathcal{X}_0$ and all $j \in \mathcal{X}_0 \cup \{1\}$, there exists $n(i,j) > 0$ such that the probability $P_{ij}^{n(i,j)}$ of hitting state $j$ in $n(i,j)$ steps starting from state $i$ for the DTMC of $\{X(t)\}$ with transition matrix $P_{ij}$, is strictly positive in equation \eqref{eqn:ctmc_trans_prob}. Let 
\begin{equation}
\epsilon(i,j) \triangleq e^{- \lambda t_2} P_{ij}^{n(i,j)} \frac{(\lambda t_1)^{n(i,j)}}{n(i,j)!}.   
\end{equation} 
Then $P_t(i,j) > \epsilon(i,j) > 0$ for all $t \in [t_1, t_2]$ where $0 < t_1 < R < T_{max} < t_2 < \infty$. Let $\epsilon_1 = \min_{i} \epsilon(i, 1) > 0$, where $i \in \mathcal{X}_0$. Therefore $P_t(i,1) \geq \epsilon_1$ for all $i$ and for all $t \in [t_1, t_2]$. Then,
\begin{equation}
\label{eqn:keyfact}
\mathbb{P}(S_{k+1} = 1, J_{k+1} = R |  S_k = x, J_k = d) \geq \epsilon_1,
\end{equation} 
for all states $(x, d)$ in the state space of the process $\{(S_k, J_k)\}$ with $x \in \mathcal{X}_0$.

We note that the inequality \eqref{eqn:keyfact} is true if we replace $(S_{k+1} = 1, J_{k+1} = R)$ by the term $(S_{k+1} = y, J_{k+1} = 2d)$ for any $y \in \mathcal{X}_0$  possibly with a different $\epsilon_1(y) > 0$ as a lower bound. We will use this fact later in the proof.

\textbf{Step 2: } 
Let us denote by $Y'$ the age of the ON period at $U_k$, i.e., the time elapsed since the ON period started, when $S_k = 1, J_k = R$ at $U_k$ and let $F_{Y'}$ be its cdf. Now
\begin{equation*}
\begin{aligned}
\sum_{y \in \mathcal{X}_0} & \mathbb{P}(S_{k+1} = y, J_{k+1} = M | S_k = 1, J_k = R) \\
& \geq \int_{s} \int_{u=0}^{R}  \mathbb{P}(Y_{on} = s + u | Y' = s)  \mathbb{P}(Y_{off} > R - u) du \hspace*{1mm} dF_{Y'}(s) \\
& \geq \int_{s} \int_{u=0}^{R}  \mathbb{P}(Y_{on} = s + u | Y_{on} \geq s) \mathbb{P}(Y_{off} > R) du \hspace*{1mm} dF_{Y'}(s) \\
& = \mathbb{P}(Y_{off} > R) \int_{s} \mathbb{P}(Y_{on} \leq s + R | Y_{on} \geq s) dF_{Y'}(s) \\
& > \epsilon_0 \mathbb{P}(Y_{off} > R)\triangleq \epsilon_0' > 0.
\end{aligned}
\end{equation*}
where the last inequality follows from \eqref{eqn:NBU} and $\mathbb{P}(Y_{off} > R) > 0$ because $Y_{off}$ is phase-type.

Thus with probability $> \epsilon_0'$ the process exits the ON period to visit the OFF period in one step. The inequality \eqref{eqn:keyfact} is true for the state $(S_{k+1} = i_0, J_{k+1} = 2M)$ with some $\epsilon_2 > 0$ as the lower bound. Therefore, $\mathbb{P}(S_{k+1} = i_0, J_{k+1} = 2M |  S_k = y, J_k = M) \geq \epsilon_2 > 0$  for any $y \in \mathcal{X}_0$. This shows that the process $\{(S_k, J_k, W_k, H_k)\}$ visits the state $(i_0, 2M, 1, 1)$ in a sequence of two steps from the ON state with probability $> \epsilon_0' \epsilon_2 > 0$, (subfigure 1 of Figure  \ref{fig:NBUON_phaseOFF_pRecurrence}) .

\begin{figure}
  \centering
  \includegraphics[scale=0.45]{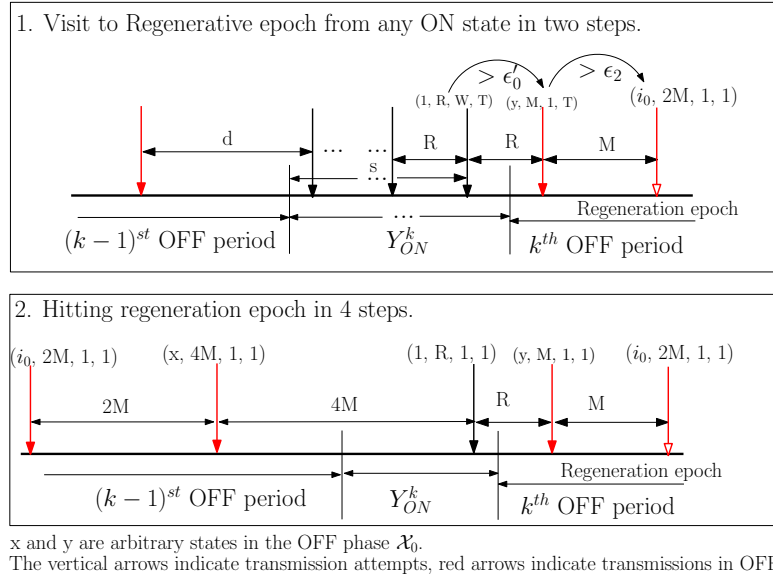}
%  \caption{The regeneration epoch $(i_0, 2M, 1, 1)$ can be hit with non-zero probability in finite time from any state in process $\{(S_k, J_k, W_k, H_k)\}$.}
  \caption[Regeneration time for phase-type OFF]{Hitting $(i_0, 2M, 1, 1)$ in finite time from any state in process.}
  \label{fig:NBUON_phaseOFF_pRecurrence}
%  \vspace*{-5mm}
\end{figure}

\textbf{Step 3: }  From steps $1$ and $2$, the process can visit the regeneration epoch 
 $(i_0, 2M, 1, 1)$ from any state in the state space in less than three steps with probability $> \epsilon =  \epsilon_0' \epsilon_1 \epsilon_2 > 0$. Thus in particular, $\mathbb{P}(N = 3) > 0$. Consider random variable $Z$ with distribution $\mathbb{P}(Z = 3k) = (1 - \epsilon)^{k-1} \epsilon , \text{ for } k \geq 1$. The random variable $Z$ is stochastically larger than the regeneration length $N$, i.e., $\mathbb{P}(Z \geq \beta) \geq \mathbb{P}(N \geq \beta)$ for all $\beta > 0$. Therefore, for all $\alpha > 0$, $\mathbb{E}[N^{\alpha}] \leq \mathbb{E}[Z^{\alpha}]  < \infty$. Also, the moment generating function of $N$ is finite in a neighbourhood of $0$.

We have shown that $\mathbb{P}(N=3) > 0$. In subfigure $2$ of Figure \ref{fig:NBUON_phaseOFF_pRecurrence}, we show a sequence of transitions such that $\mathbb{P}(N = 4) > 0$. This shows that $N$ is aperiodic. Thus  $\{(S_k, J_k, W_k, H_k)\}$ has a unique stationary distribution, \eqref{eqn:regen_eqn_phaseOFF} and converges to it in total variation from any initial distribution exponentially (because of finiteness of mgf of $N$ in a neighbourhood of $0$ \cite{Kalashnikov1993}). 
\end{proof}

The condition \eqref{eqn:NBU} is satisfied by a general class of distributions called New Better than Used (NBU)\cite{Marshall2007}  and also by phase-type distributions. The NBU distributions are useful in reliability theory. They are also relevant in our case, as one would typically expect that the primary busy period starting afresh is likely to last longer than an ongoing busy period.

\subsection{Phase-type ON Periods}
\label{sec:generalOFF_Extension}
We now develop theoretical results for phase-type ON and general OFF periods. Let $X(t)$ denote the phase of the ON period. The process $\{X(t)\}$ is a CTMC with a finite state space $\{0\} \cup \mathcal{X}_1$ with $0 \notin \mathcal{X}_1$. Let us denote by $S(t)$ the state of the ON-OFF channel at time $t$. If the channel is ON at  time $t$, then $S(t) = X(t)$, else $S(t) = 0$, where $0$ is the absorbing state for the CTMC $\{X(t)\}$. We assume $\mathcal{X}_1$ is irreducible and $0$ is reachable from all states. The process $\{(S_k, J_k, W_k, H_k)\}$ is regenerative with visits to the state $(i_1, R, 1, 1)$, ($i_1 \in \mathcal{X}_1$ is a fixed state), acting as regeneration epochs. We denote by $N$ the number of transmission attempts made in a regeneration cycle. 

\begin{proposition}
\label{thm:phaseON_stationary}
If the OFF periods are i.i.d. with $\mathbb{P}(s + M < Y_{off} < s + 3M) \geq \epsilon_0 > 0$ for all $s \leq R$, and 
\begin{equation} 
\label{eqn:NBUOFF} P(Y_{off} \leq s + 2M| Y_{off} \geq s) \geq \epsilon_1 > 0,
\end{equation}
then for all $\alpha > 0$, $\mathbb{E}N^{\alpha}] < \infty$ and $N$ has a finite  moment generating function in a neighbourhood of $0$. Also, $N$ is aperiodic and the stochastic process $\{(S_k, J_k, W_k, H_k)\}$ converges in total variation, exponentially to its unique stationary distribution \eqref{eqn:regen_eqn_phaseOFF}.
%\begin{equation} 
%\label{eqn:NBUOFF} \mathbb{P}(Y_{off} \leq s + d| Y_{off} \geq s) \geq \epsilon_1 > 0, 
%\end{equation}
%for $d \in \{2M, 4M, 8M, \cdots T_{max} \}$. Then for all $\alpha > 0$, $\mathbb{E}[N^{\alpha}] < \infty$ and $N$ has a finite  moment generating function in a neighbourhood of $0$. Also, $N$ is aperiodic and the stochastic process $\{(S_k, J_k, W_k, H_k)\}$ converges in total variation, exponentially to its unique stationary distribution \eqref{eqn:regen_eqn_phaseOFF}.
\end{proposition}
\begin{proof} We first prove that from any arbitrary state in the ON period the process can visit the OFF state in one step with probability $\geq \epsilon_2 > 0$.

\textbf{Step 1: } Since $\mathcal{X}_1$ is irreducible and $0$ is reachable from all states, for all $i \in \mathcal{X}_1$ and all $j \in \mathcal{X}_1 \cup \{0\}$, there exists $n(i,j) > 0$ such that the probability $P_{ij}^{n(i,j)}$ of hitting state $j$ in $n(i,j)$ steps starting from state $i$ for the DTMC of $\{X(t)\}$ with transition matrix $P_{ij}$, is strictly positive in equation \eqref{eqn:ctmc_trans_prob}. Let 
\begin{equation}
\epsilon(i,j) \triangleq e^{- \lambda R} P_{ij}^{n(i,j)} \frac{(\lambda R)^{n(i,j)}}{n(i,j)!}.   
\end{equation} 
Then $P_R(i,j) > \epsilon(i,j) > 0$. Let $\epsilon_2 = \min_{i} \epsilon(i, 0) > 0$, where $i \in \mathcal{X}_1$. Therefore $P_R(i,0) \geq \epsilon_2$. Then,
\begin{equation}
\label{eqn:keyfact_2}
\mathbb{P}(S_{k+1} = 0, J_{k+1} = M |  S_k = x, J_k = R) \geq \epsilon_2,
\end{equation} 
for all $x \in \mathcal{X}_0$. 

\textbf{Step 2: } 
Let us denote by $Y'$ the age of the OFF period at $U_k$, i.e., the time elapsed since the OFF period started, when $S_k = 0, J_k = d$ at $U_k$ and let $F_{Y'}$ be its cdf. Now, for some $d \in \{2M, 4M, 8M, \cdots T_{max} \}$,
\begin{equation*}
\begin{aligned}
\sum_{y \in \mathcal{X}_1} & \mathbb{P}(S_{k+1} = y, J_{k+1} = R | S_k = 0, J_k = d) \\
& \geq \int_{s} \int_{u=0}^{d}  \mathbb{P}(Y_{off} = s + u | Y' = s)  \mathbb{P}(Y_{on} > d - u) du \hspace*{1mm} dF_{Y'}(s) \\
& \geq \int_{s} \int_{u=0}^{d}  \mathbb{P}(Y_{off} = s + u | Y_{off} \geq s) \mathbb{P}(Y_{on} > d) du \hspace*{1mm} dF_{Y'}(s) \\
& = \mathbb{P}(Y_{on} > d) \int_{s} \mathbb{P}(Y_{off} \leq s + d | Y_{off} \geq s) dF_{Y'}(s) \\
& > \epsilon_1 \mathbb{P}(Y_{on} > d)\triangleq \epsilon_1' > 0.
\end{aligned}
\end{equation*}
where the last inequality follows from \eqref{eqn:NBUOFF} and $\mathbb{P}(Y_{on} > d) \geq \mathbb{P}(Y_{on} > T_{max}) > 0$, for $d \in \{2M, 4M, 8M, \cdots T_{max} \}$, because $Y_{on}$ is phase-type.

Thus with probability $> \epsilon_1'$ the process exits the OFF period to visit the ON period in one step. 

\textbf{Step 3: } Let us denote by $Y'(z)$ the residual life time of the ON period at $U_k$ when $S_k = z \in \mathcal{X}_1, J_k = R$ at $U_k$ and let $F_{Y_{z}'}$ be its cdf. For any state $z \in \mathcal{X}_1$, $\mathbb{P}(Y'(z) < R) > 0$. Therefore, $\min_{z \in \mathcal{X}_1} \mathbb{P}(Y'(z) < R) > 0$. Now
\begin{equation}
\begin{aligned}
\displaystyle{\sum_{y \in \mathcal{X}_1}} \mathbb{P}(S_{k+3} & = y, J_{k+3} = R, W_{k+3} = 1, H_{k+3} = 1 | S_k = z, J_k = R) \\
& \geq \int_{u=0}^{R} \mathbb{P}(R - u + M < Y_{off} < R - u + 3M) \mathbb{P}(Y_{on} > 4M) dF_{Y_{z}'}(u) \\
& \geq \epsilon_0 \mathbb{P}(Y_{on} > 4M) \int_{u=0}^{R}   dF_{Y_{z}'}(u) \\
& \geq \epsilon_0 \mathbb{P}(Y_{on} > 4M) \mathbb{P}(Y'(z) < R) \\
& \triangleq \epsilon_0' > 0. 
\end{aligned}
\end{equation}
for all $z \in \mathcal{X}_1$. Therefore there exists some $i_1 \in \mathcal{X}_1$ such that $ \mathbb{P}(S_{k+3} = i_1, J_{k+3} = R, W_{k+3} = 1, H_{k+3} = 1 | S_k = z, J_k = R) \geq \epsilon_2' > 0$ for all $z \in \mathcal{X}_1$. Therefore, there is a positive probability of the process 
hitting state $(i_1, R, 1, 1)$ from any ON state in exactly three steps.

\textbf{Step 4: } From above steps, we see that the process can visit the regenerative epoch $(i_1, R, 1, 1)$ from any state in the state space in less than four steps with probability $\delta = \epsilon_1'\epsilon_2'$. In particular, $N$ satisfies $\mathbb{P}(N = 4) > 0$. Consider random variable $Z$ with distribution $\mathbb{P}(Z = 4k) = (1 - \delta)^{k-1} \delta , \text{ for } k \geq 1$. The random variable $Z$ is stochastically larger than the regeneration length $N$, i.e., $\mathbb{P}(Z \geq \beta) \geq \mathbb{P}(N \geq \beta)$ for all $\beta > 0$. Therefore, for all $\alpha > 0$, $\mathbb{E}[N^{\alpha}] \leq \mathbb{E}[Z^{\alpha}]  < \infty$. Also, the moment generating function of $N$ is finite in a neighbourhood of $0$.

In Step $3$, we have shown that $\mathbb{P}(N = 3) > 0$. In subfigure $2$ of Figure \ref{fig:genOFF_phaseON_pRecurrence}, we show a sequence of transitions such that $\mathbb{P}(N = 4) > 0$. This shows that $N$ is aperiodic. Thus  $\{(S_k, J_k, W_k, H_k)\}$ has a unique stationary distribution, \eqref{eqn:regen_eqn_phaseOFF} and converges to it in total variation from any initial distribution exponentially (because of finiteness of mgf of $N$ in a neighbourhood of $0$ \cite{Kalashnikov1993}). 

\begin{figure}
  \centering
  \includegraphics[scale=0.45]{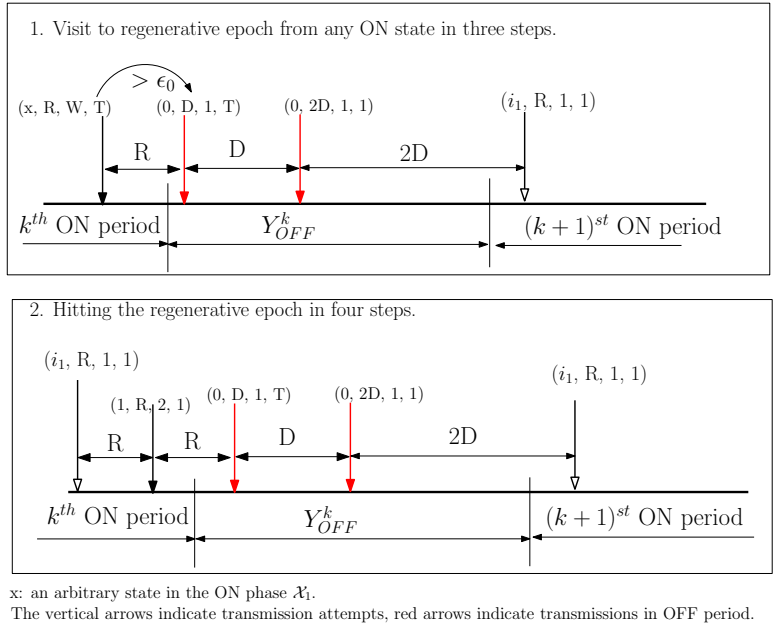}
  %\caption{The regeneration epoch $(i_1, R, 1, 1)$ can be hit with non-zero probability in finite time from any state in process $\{(S_k, J_k, X_k, W_k, H_k)\}$.}
  \caption[Regeneration time for phase-type ON]{Hitting $(i_1, R, 1, 1)$ in finite time from any state in process.}
  \label{fig:genOFF_phaseON_pRecurrence}
\end{figure}

\end{proof}

The conditions in Proposition \eqref{thm:phaseON_stationary} are satisfied if $Y_{off}$ has a positive density on $[M,R+3M]$ and has NBU distribution. It is also satisfied for phase-type  distributions. 

In the proofs for Propositions \eqref{thm:phaseOFF_stationary} and \eqref{thm:phaseON_stationary}, we do not assume random packet losses for convenience of notation. The propositions hold even in this case when the packet loss probability, $p < 1$ with a slight modification of proofs. 

These propositions show stationarity of the regenerative process modelling TCP behaviour in the setup of ON-OFF channels with more general ON and OFF distributions. This also ensures that the time averages for performance metrics such as throughput and probability of RTO converge to the stationary mean values.

%We need to add a condition in the proofs that no packets are lost in the ON periods that we consider. The window size $W \leq W_{max}$, where $W_{max}$ is the maximum window size of TCP. Therefore the event of no packet losses happens with probability $> (1 - p)^{W_{max}} > 0$. Thus the propositions \ref{thm:phaseOFF_stationary} and \ref{thm:phaseON_stationary} help us to extend analysis of TCP performance for a wider class of ON and OFF distributions. 
%In the next section, we show the simulation results with different ON and OFF periods and compare them with results obtained using the theoretical models. 

\subsection{Simulation Results}
\label{sec:simulationResults_genBusy}
We now compare the probability of timeout and the throughput obtained via the analytical model with ns2 simulations. The probability of timeout and throughput can be computed using equations \eqref{eqn:PKT_RTO} and \eqref{eqn:T_palmcalculus} respectively with some modifications. For phase-type ON, we replace the term $1_{\{S = 1\}}$ by $1_{\{S \in \mathcal{X}_1\}}$ and for phase-type OFF process we replace $1_{\{S = 0\}}$ by $1_{\{S \in \mathcal{X}_0\}}$. The simulation setup is the same as in Section \ref{sec:simulationResults}.

%We consider four different cases: (a) ON and OFF periods are both exponentially distributed, (b) ON and OFF are both Erlang-2, (c) ON and OFF are both Erlang-3 and (d) ON period is exponential whereas OFF period is uniform. In Figure \ref{fig:BX_IX_pkt_rto}, we compare the probability of RTO calculated by our theoretical model with ns2 simulations for these cases.  In Figure \ref{fig:BX_IX_goodput}, we compare the throughput computed by our theoretical model with ns2 simulations. For these experiments, we set the RTT to $0.1$ sec and $W_{max}$ to $100$ and the ON-OFF channel link speed is set to $5$ Mbps. The packets undergo Bernoulli random losses with probability $0.01$. We vary the average busy duration, $EB$ keeping $\alpha$ fixed at $1/3$.  We see that our theoretical model results match well with simulation with errors in throughput and RTO less than $5\%$ in most cases. The errors are larger when the average ON and OFF durations are of the order of RTT. However, even for these cases the errors are less than $13\%$. 

We consider two cases (a) ON and OFF periods are both exponentially distributed and (b) ON and OFF are both Erlang-3 distributed. For these experiments, we set RTT to $0.1$ sec, $W_{max}$ to $100$ and the ON-OFF channel link speed is set to $5$ Mbps. The packets undergo Bernoulli random losses with probability $0.01$. We vary the average busy duration, $\mathbb{E}[Y_{off}]$ keeping $\alpha$ fixed at $1/3$. The results are shown in 
Figure \ref{fig:BX_IX}.
%Figure \ref{fig:BX_IX_pkt_rto}, we compare the probability of RTO calculated by our theoretical model with ns2 simulations for these cases.  In Figure \ref{fig:BX_IX_goodput}, we compare the throughput computed by our theoretical model with ns2 simulations. 
We see that our theoretical model results match well with simulations with errors less than $5\%$ in most cases. The errors are larger when the average ON and OFF durations are of the order of RTT. However, even for these cases the errors are less than $10\%$. 

In Figure \ref{fig:BX_IX_5M_1M}, we consider the effect of the cognitive channel link capacity on probability of timeout and throughput. The ON and OFF periods are both Erlang-3 distributed, we set RTT to $0.1$ seconds and $W_{max}$ to $100$. We consider link speeds of $1$ Mbps and  $5$ Mbps. We see that our theoretical model results match well with simulations with errors less than $5\%$ in most cases and always less than $11\%$ .

%\begin{figure}
%  \centering
%  \includegraphics[scale=0.15]{BX_IX_pkt_rto.eps}
%  \caption{Probability of RTO with different ON and OFF distributions}
%  \label{fig:BX_IX_pkt_rto}
%\end{figure}   
%
%\begin{figure}
%  \centering
%  \includegraphics[scale=0.16]{BX_IX_goodput.eps}
%  \caption{Secondary TCP throughput with different ON and OFF distributions}
%  \label{fig:BX_IX_goodput}
%\end{figure}  

\begin{figure}
\centering
\begin{tabular}{c}
\includegraphics[scale=0.25, trim = 70 0 140 5, clip=true]{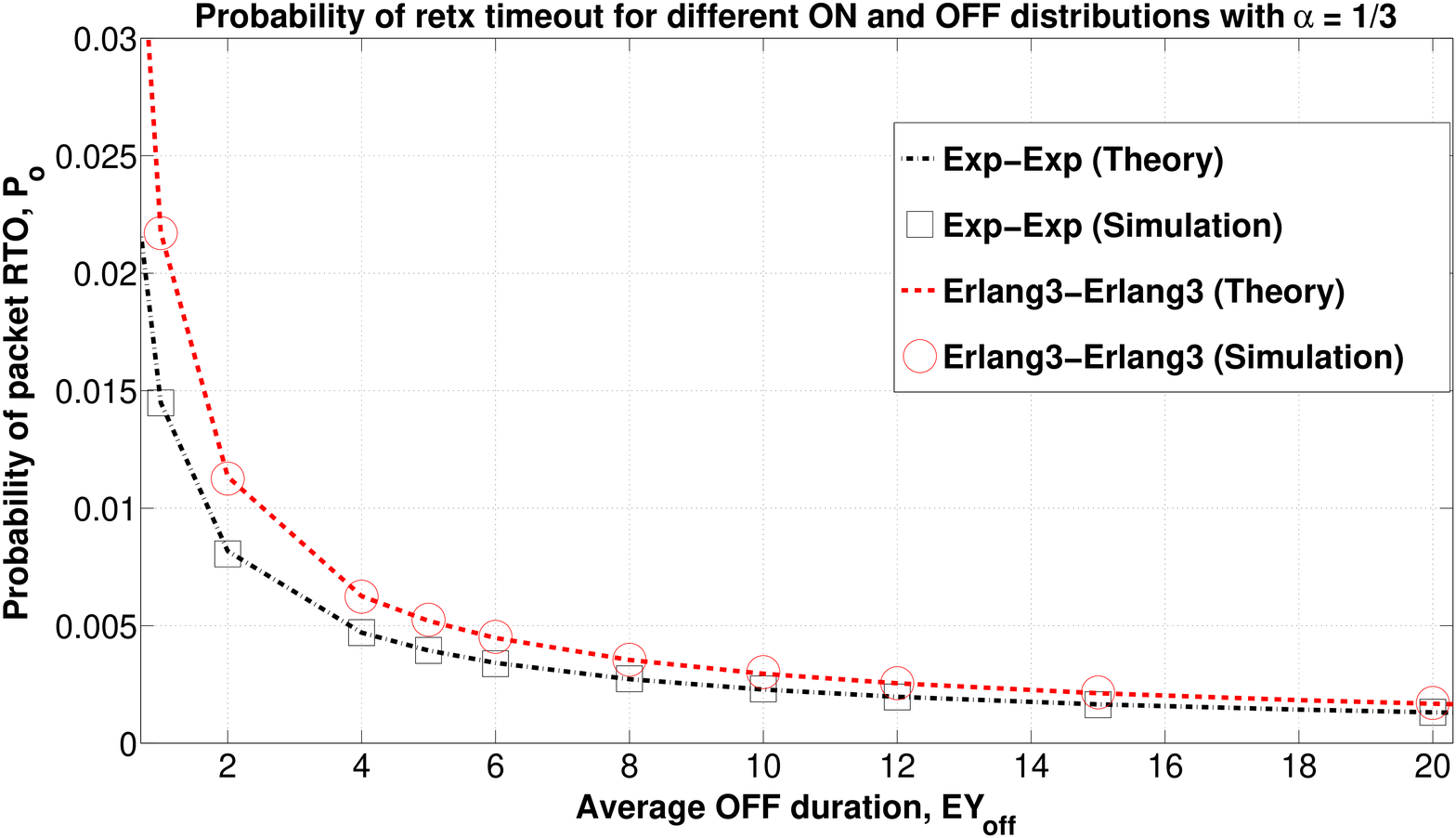} \\ \\ \\
\includegraphics[scale=0.25, trim = 70 0 140 5, clip=true]{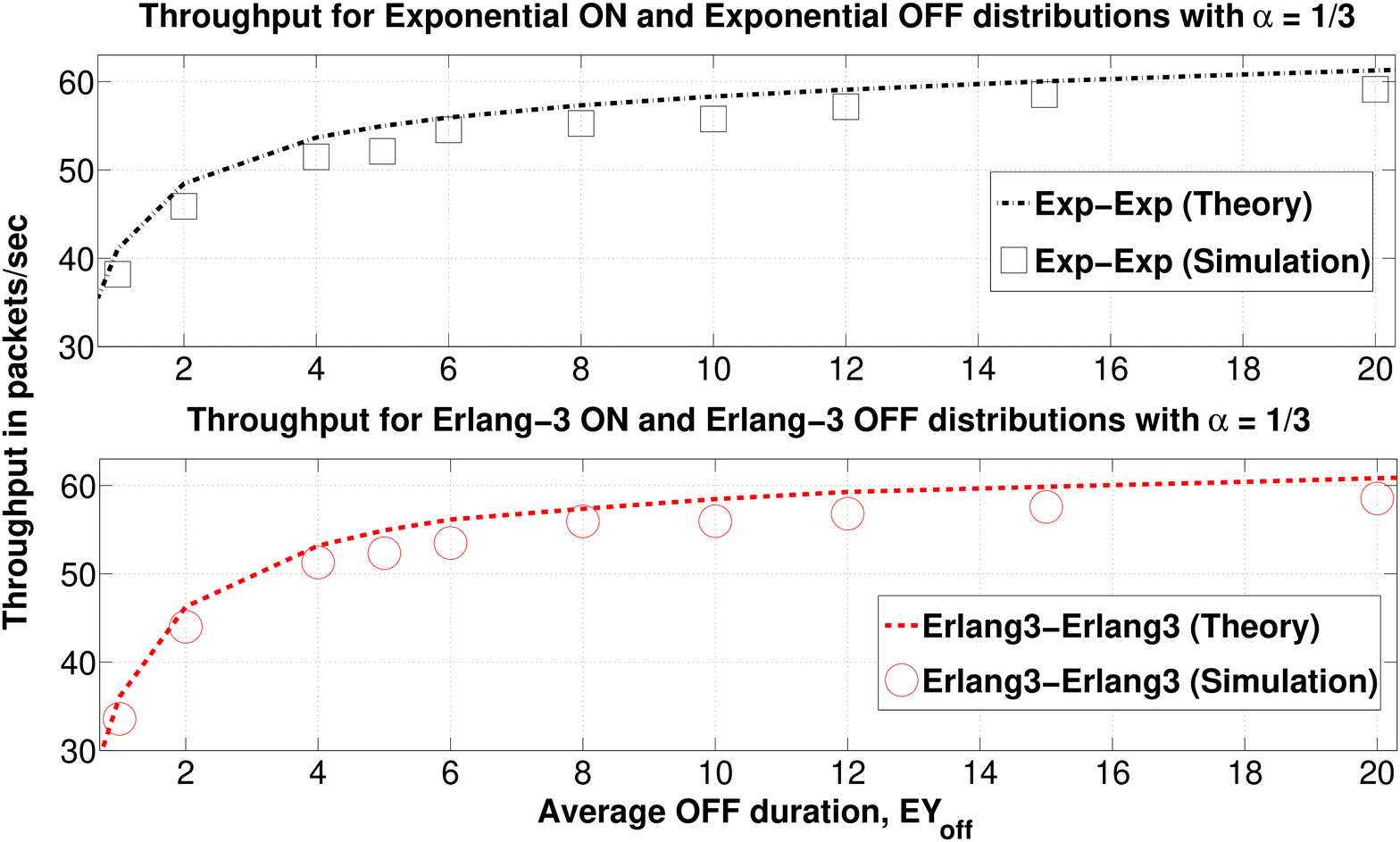}
\end{tabular}
\caption{Probability of RTO, throughput with different ON-OFF distributions.}
\label{fig:BX_IX}
\end{figure}

\begin{figure}
\centering
\begin{tabular}{c}
\includegraphics[scale=0.25, trim = 70 0 140 5, clip=true]{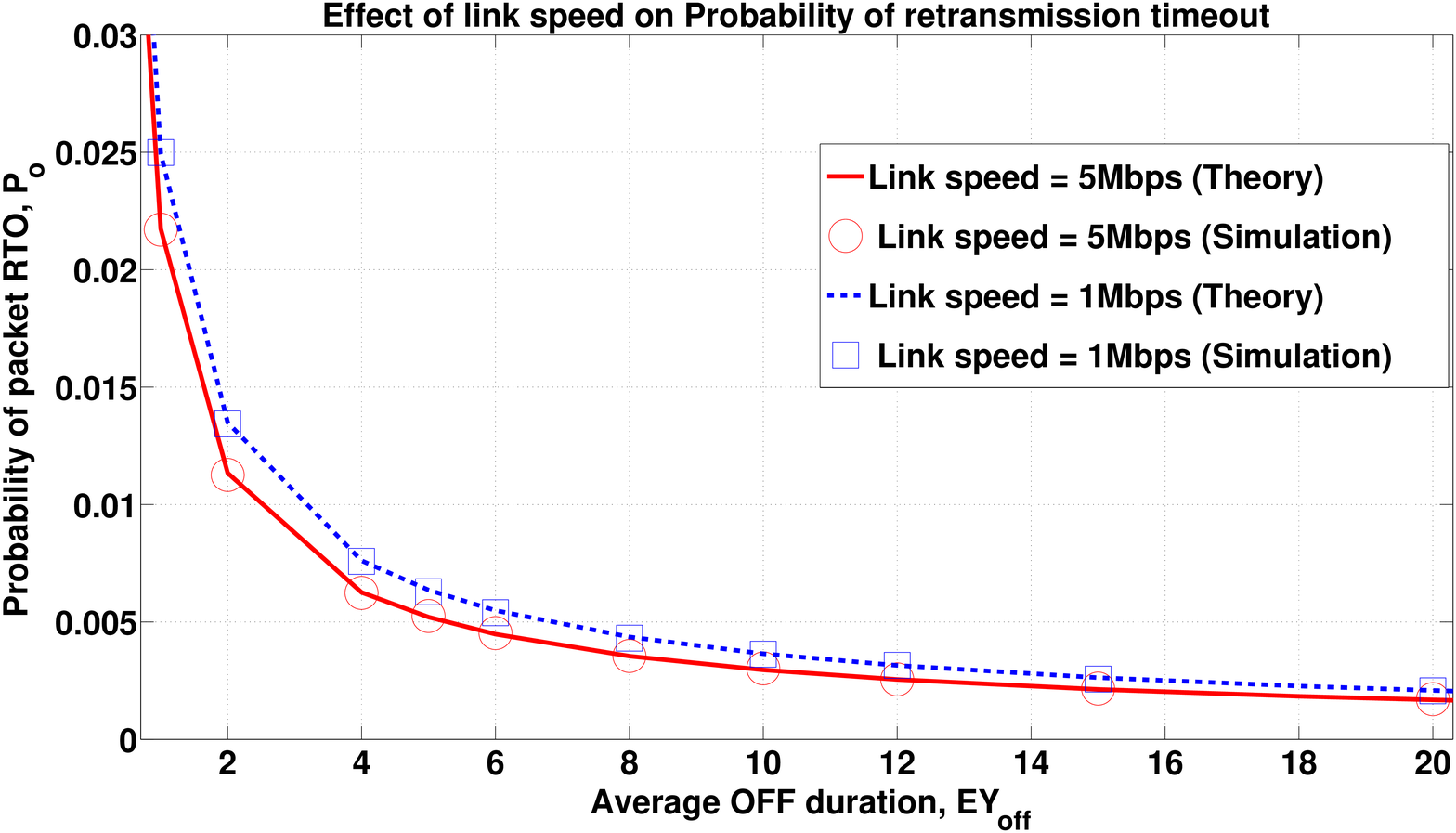}  \\ \\ \\
\includegraphics[scale=0.25, trim = 70 0 140 5, clip=true]{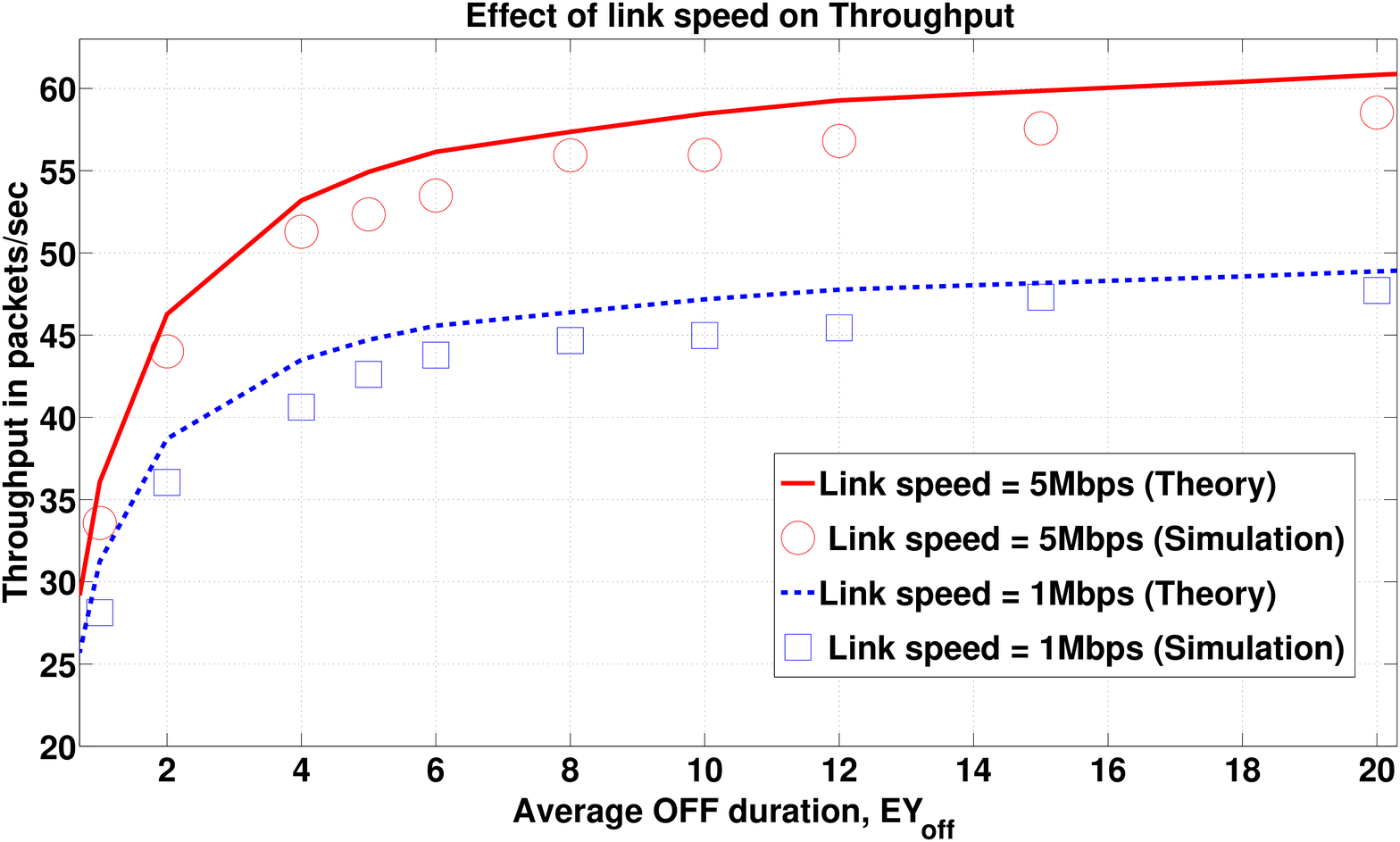}
\end{tabular}
\caption{Probability of RTO, throughput with different link speeds.}
\label{fig:BX_IX_5M_1M}
%\vspace*{-5mm}
\end{figure}
\section{Multiple TCP Connections}
\label{sec:multipleTCP}
In this section, we consider the case when multiple secondary TCP connections share a CR channel. Thus the ON-OFF durations for all connections are same. However these connections are subject to different packet error rates as their channel gains may be different. Also these connections can possibly go through different routes; therefore they may have different round trip times. We first consider the case where the queuing delays are negligible. Then, the processes $\{(S_k, J_k, W_k, H_k)\}$ for each TCP can be considered individually with no effect on each other. The probability of timeout and throughput can be computed using equations \eqref{eqn:PKT_RTO} and \eqref{eqn:T_palmcalculus} respectively.  The following simulations show that this model works fine for system parameters considered here.

We consider $6$ secondary TCP connections with different RTTs and packet loss probabilities sharing an ON-OFF channel with Erlang-2 distributed ON and OFF periods with average ON duration = $20$ seconds and average OFF duration = $10$ seconds. The ON-OFF channel has link speed of $10$ Mbps and the other links are set to $1$ Gbps. The TCP packet sizes are $1050$ bytes. %We compare the throughput results from ns2 simulations with theoretical results in Table \ref{tbl:multipleTCP}. 
The packet error rates of the different TCP connections are as given in Table \ref{tbl:multipleTCP}. We compare the throughput and probability of timeout, $P_o$ obtained using ns2 simulations with theoretical results in Table \ref{tbl:multipleTCP}. 
The difference between the simulation results and analytical model results is less than $9\%$. 

%\begin{table}
%\caption{Multiple TCP connections}
%\begin{tabular}{|c|c|c|c|c|}
%\hline
%Flow index i & $PER_i$ & $RTT_i$ & Throughput & Throughput \\
% &  & (sec) & (ns2) &  (Theoretical) \\
%\hline
%1& $0.01$   &	$0.05$ &	 $108.18$ & $116.00$ \\
%2& $0.01$   & 	$0.10$ &	 $54.14$ & $57.34$ \\
%3& $0.01$   &	$0.20$ &	 $26.02$ & $27.37$ \\
%4& $0.005$  &	$0.05$ &	 $154.44$ & $164.34$ \\
%5& $0.005$  &	$0.10$ &	 $77.00$ & $80.21$ \\
%6& $0.005$  &	$0.20$ &	 $36.40$  & $37.54$ \\
%\hline
%\end{tabular}
%\label{tbl:multipleTCP}
%\vspace*{-5mm}
%\end{table}

\begin{table}
\centering
\caption{Multiple TCP flows with negligible queuing.}
\begin{tabular}{|c|c|c|c|c|c|}
\hline
$PER_i$ & $RTT_i$ & Throughput & Throughput & $P_o$ & $P_o$ \\
& (sec) & (ns2) &  (Theoretical) & (ns2) &  (Theoretical) \\
\hline
$0.01$   &	$0.05$ &	 $108.9$ & $118.6$ & $0.00148$ & $0.00136$\\
$0.01$   & 	$0.10$ &	 $54.9$ & $58.5$ & $0.00291$ & $0.00274$\\
$0.01$   &	$0.20$ &	 $26.5$ & $28.2$ & $0.00591$ & $0.00566$\\
$0.005$  &	$0.05$ &	 $156.8$ & $167.4$ & $0.00103$ & $0.00097$\\
$0.005$  &	$0.10$ &	 $78.9$ & $81.8$ & $0.00205$ & $0.00197$\\
$0.005$  &	$0.20$ &	 $37.4$  & $38.6$ & $0.00425$ & $0.00416$\\
\hline
\end{tabular}
\label{tbl:multipleTCP}
\end{table}

When the network has non-negligible queuing and all the flows have same RTT, we can extend the model from Section \ref{sec:phase_type}. Suppose there are $N$ TCP flows in the network and they share the ON-OFF channel. If the ON and OFF periods are phase-type. then $ (S_k, J_k, (W_k^j, H_k^j)_{j \in \{1, 2, \cdots, N\}})$ denotes the state of the system and forms a Markov chain, where index $j$ represents the TCP connection $j$. The RTT for the different flows at the end of the $k^{th}$ RTT is given by $\max\{ \Delta, \frac{\sum_j {W_k^j}}{\mu}\}$ where $\mu$ is the bottleneck link speed (in packets/sec) and $\Delta$ is the propagation delay (in sec). Our simulation results validate our model assumptions.

We consider $3$ secondary TCP connections with $\Delta = 0.1$ seconds. The maximum window size for all the flows is set to $20$ packets. The ON-OFF periods are both exponentially distributed with average ON duration = $20$ seconds and average OFF duration = $10$ seconds. The ON-OFF channel has link speed of $2$ Mbps (this causes non-negligible queuing) and the other links are set to $1$ Gbps. The TCP packet sizes are $1050$ bytes. The packet error rates of the different TCP connections are as given in Table \ref{tbl:multipleTCP_2}.  We compare the throughput and probability of timeout, $P_o$ obtained using ns2 simulations with theoretical results in Table \ref{tbl:multipleTCP_2}. In this case, the errors are less than $13\%$.

\begin{table}
\centering
\caption{Multiple TCP flows with non-negligible queuing.}
\begin{tabular}{|c|c|c|c|c|c|}
\hline
$PER_i$ & $RTT_i$ &  Throughput & Throughput & $P_o$ & $P_o$ \\
& (sec) & (ns2) &  (Theoretical) & (ns2) &  (Theoretical) \\
\hline
$0.01$ & $0.1$  &	 $30.83$ & $28.19$ & $0.00404$ & $0.00455$\\
$0.003$ & $0.1$   &	 $42.06$ & $38.51$ & $0.00320$ & $0.00336$\\
$0.001$ & $0.1$  &	 $47.16$ & $42.80$ & $0.00296$ & $0.00303$\\
\hline
\end{tabular}
\label{tbl:multipleTCP_2}
%\vspace*{-5mm}
\end{table}

\section{Conclusions}
\label{sec:conclusion}
We have developed an analytical Markov model for a TCP flow over an ON-OFF channel with random losses. For the Markov model, we assume that the ON and OFF periods are both exponential. We then extend our model to include phase-type ON and phase-type OFF periods. We have also considered the case with more general ON and OFF periods and proved the stationarity of the system using regenerative process theory. We have compared the results viz., probability of retransmission timeout and secondary TCP throughput obtained using the theoretical models with ns2 simulations and showed that these match quite well. Finally, we considered the scenario where multiple secondary TCP connections share the ON-OFF channel. The theoretical results for this scenario also match well with simulations.
\bibliographystyle{IEEEtran} 
\bibliography{tcp-cognitive}
\end{document}